\documentclass[%
    superscriptaddress,%
    amsmath,%
    amssymb,%
    aps,
    pra,%
    twocolumn,%
    floatfix,%
]{revtex4-2}
\usepackage{graphicx}
\usepackage{bm,bbm}
\usepackage[colorlinks=true,
bookmarks=false,
linkcolor=blue,
urlcolor=blue,
citecolor=blue,
breaklinks]{hyperref}
\usepackage{xcolor}
\usepackage{physics}
\usepackage{comment}
\usepackage{amsthm}
\usepackage{mathtools}
\usepackage[ruled,vlined]{algorithm2e}
\newtheorem{theorem}{Theorem}[section]
\newtheorem{definition}{Definition}[section]

\newtheorem{lemma}[theorem]{Lemma}
\usepackage{color}
\usepackage{rotating}
\usepackage{tabularray}
\usepackage{hhline}
\usepackage{booktabs}

\begin{document}
\title{Treespilation: Architecture- and State-Optimised Fermion-to-Qubit Mappings}

\author{Aaron Miller}
\email{aaron.miller@algorithmiq.fi}
\affiliation{Algorithmiq Ltd, Kanavakatu 3C 00160 Helsinki, Finland}
\affiliation{Trinity Quantum Alliance, Unit 16, Trinity Technology and Enterprise Centre, Pearse Street, Dublin 2, Ireland}
\author{Adam Glos}
\affiliation{Algorithmiq Ltd, Kanavakatu 3C 00160 Helsinki, Finland}
\author{Zolt\'an Zimbor\'as}
\affiliation{Algorithmiq Ltd, Kanavakatu 3C 00160 Helsinki, Finland}
\affiliation{HUNREN Wigner Research Centre for Physics, Budapest, Hungary}

\date{\today}
\begin{abstract}
Quantum computers hold great promise for efficiently simulating Fermionic systems, benefiting fields like quantum chemistry and materials science.
To achieve this, algorithms typically begin by choosing a Fermion-to-qubit mapping to encode the Fermionic problem in the qubits of a quantum computer.
In this work, we introduce "treespilation," a technique for efficiently mapping Fermionic systems using a large family of favourable tree-based mappings previously introduced by some of the authors.
We use this technique to minimise the number of CNOT gates required to simulate chemical groundstates found numerically using the ADAPT-VQE algorithm.
We observe significant reductions, up to $74\%$, in CNOT counts on full connectivity and for limited qubit connectivity-type devices such as IBM Eagle and Google Sycamore, we observe similar reductions in CNOT counts. 
In many instances, the reductions achieved on these limited connectivity devices even surpass the initial full connectivity CNOT count.
Additionally, we find our method improves the CNOT and parameter efficiency of QEB- and qubit-ADAPT-VQE, which are, to our knowledge, the most CNOT-efficient VQE protocols for molecular state preparation.
\end{abstract}

\maketitle

\section{Introduction}
Quantum computing has made significant strides in the past decade.
However, achieving fault-tolerant quantum computing remains a challenging goal. 
Current quantum devices have limitations such as a small number of qubits, restricted qubit connectivity, and error-prone gates, making it difficult to execute deep circuits required for paradigmatic quantum algorithms \cite{preskill2018quantum}. 
Nevertheless, recent experiments have demonstrated the potential of today's quantum devices, and have shown success in solving complex problems \cite{IBMexperiment, googleexperiment, rosenberg2023dynamics, Nathansim}. 
This potential offers valuable computational resources, particularly when combined with classical computing \cite{arute2019quantum, zhong2020quantum, wu2021strong, madsen2022quantum}, especially in mitigating the detrimental effects of noise \cite{filippov2023scalable, endo2018practical, endo2021hybrid, cai2022quantum, berg2022probabilistic}.

Among the diverse applications of quantum computing, simulating many-body Fermionic quantum systems with quantum devices presents an intriguing prospect, especially in computational chemistry \cite{kandala2017hardware,barkoutsos2018quantum,ollitrault2020quantum}. 
This potential transformation extends to fields like material science \cite{lordi2021advances} and drug discovery \cite{cao2018potential,blunt2022perspective,maniscalco2022quantum}, among others.
Various approaches exist for addressing these quantum chemical problems on quantum computers \cite{fitzpatrick2022selfconsistent, kirby_exact_2023, nykanen_toward_2023}, with many utilizing the physical qubit state of the quantum device to represent the desired many-body Fermionic system.
Properties of the system are then inferred through measurements of the qubit state \cite{mcardle2020quantum, tilly2022variational}.
One such approach is the Variational Quantum Eigensolver (VQE) \cite{cerezo_variational_2022}, which approximates the qubit representation of a target Fermionic state, such as the ground state of a Fermionic Hamiltonian. 
The algorithm begins by deriving a qubit Hamiltonian from the desired Fermionic Hamiltonian using a Fermion-to-qubit (F2Q) mapping. 
Next, a parameterized quantum circuit, known as an ansatz, is designed. 
Finally, the circuit parameters are optimized using a classical optimizer, to minimize the energy of the current quantum state for the qubit Hamiltonian.

\begin{figure*}[ht!]
   \centering
   \includegraphics[width=2\columnwidth]{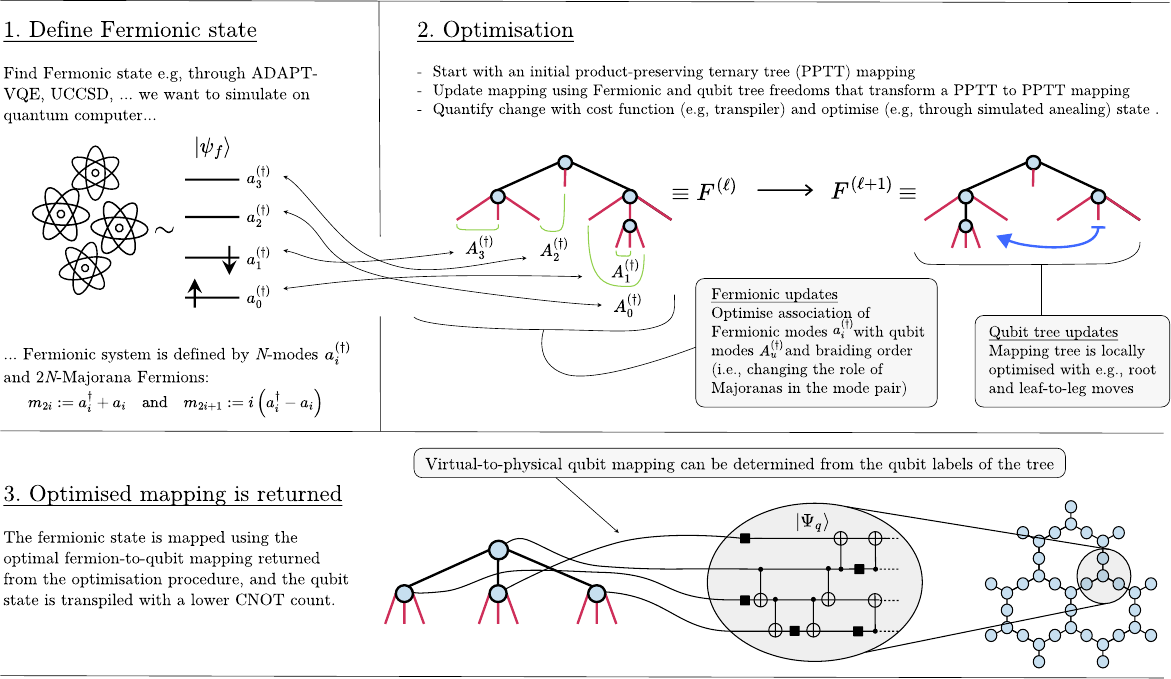}
\caption{Infographic of the treespilation algorithm.}
   \label{fig:f2q opt procedure}
\end{figure*}

For near-term quantum devices, circuit noise robustness is crucial. VQE offers the potential to discover such circuits, characterized by a reduced presence of noisy entangling two-qubit CNOT gates compared to far-term approaches like qubitisation \cite{qubitisation}.
The reduction of these gates is important as they take longer and have lower fidelities compared to single-qubit gates, contributing to computation errors \cite{gateerrors}.
One promising variant of VQE is the Adaptive Derivative-Assembled Pseudo-Trotter (ADAPT) VQE algorithm. 
It starts with a reference state, like the Hartree-Fock state, and sequentially adds elements from a predefined candidate gate set, known as a pool, to optimize for the target state \cite{adapt-vqe-orig-paper}. 
The choice of the operator pool significantly affects the convergence and circuit cost in qubit space.
Often, pools originating from Fermionic systems are chosen to produce such circuits \cite{adapt-vqe-orig-paper,yordanov2021qubit,tang2021qubit}.
The \textit{Fermionic} pool \cite{adapt-vqe-orig-paper} consisting of single- and double-excitation operations present in the Unitary Coupled Cluster Singles and Doubles (UCCSD) ansatz \cite{cao_quantum_2019}.
However, mapping Fermionic operators to qubits can result in highly non-local operations incurred from mapping indistinguishable Fermions to distinguishable qubits.
To address this challenge, the Qubit-Excitation-Based (QEB) pool was introduced, which modifies elements of the Fermionic pool to disregard Fermionic antisymmetry. 
This enables efficient implementation with a fixed number of CNOT gates for full connectivity, making it a leading method for state preparation \cite{yordanov2021qubit,burton_exact_2023,feniou_overlap-adapt-vqe_2023}. 
Another approach, the qubit-pool, reduces CNOT gate requirements further by splitting QEB pool elements into individual 4-local Pauli strings \cite{tang2021qubit}.
The unitaries in the non-Fermionic pools do not have a straightforward representation in Fermionic space.
Although a representation does exist, we refer to them as non-Fermionic pools.
Additionally, some entangler-circuit approaches aim to minimize gate count by avoiding Fermionic operations altogether \cite{PhysicsConstrainedHardwareEfficientAnsatz}.

With most approaches to solving Fermionic problems on quantum computers, a Fermion-to-qubit mapping is selected.
The mapping encodes a many-mode Fermionic Hamiltonian and target state, $\ket{\Psi_\mathrm{f}}$, as a multi-qubit Hamiltonian of Pauli operators and state, $\ket{\Psi_\mathrm{q}}$. 
The choice of mapping is not unique, and different mappings result in different qubit states with varying challenges in preparation of $\ket{\Psi_\mathrm{q}}$ on the quantum device. 
Moreover, the interest lies not only in simulating $\ket{\Psi_\mathrm{f}}$ but also in determining physical properties, $\bra{\Psi_\mathrm{f}}\mathcal{O}_\mathrm{f}\ket{\Psi_\mathrm{f}}$, of certain Fermionic observable operators $\mathcal{O}_\mathrm{f}$. 
The chosen Fermion-to-qubit mapping maps $\mathcal{O}_\mathrm{f}$ to its qubit counterpart, $\mathcal{O}_\mathrm{q}$, and the evaluation of $\bra{\Psi_\mathrm{q}} \mathcal{O}_\mathrm{q} \ket{\Psi_\mathrm{q}}$ involves measurements on a physical quantum device, incurring measurement costs depending on the chosen mapping \cite{garcia2021learning,glos2022adaptive}. 

A significant obstacle in implementing these Fermionic operations is the connectivity of the quantum device we use to simulate the state.
Limited connectivity devices, such as those based on superconducting qubits, can incur large circuit overheads when compared to full connectivity due to the necessity of SWAP gates needed to transpile the circuit to the device and due to the non-local nature of the mapped Fermionic operations. 
To address this issue, \cite{miller2022bonsai} introduces a versatile class of mappings and presents the Bonsai algorithm. 
This algorithm tailors the Fermion-to-qubit mapping to the device, reducing SWAP gate overhead by aligning the mapping's tree structure with the qubit connectivity graph.
Subsequent research has built on this approach by employing the framework to encode double excitations within two-qubit subspaces to simplify the entanglement structure and lower the computational cost of VQEs and tensor-network simulations for chemical systems \cite{parella-dilme_reducing_2023}.

The significance of Fermion-to-qubit mappings has thus fueled extensive research toward designing mappings beyond the traditional Jordan-Wigner (JW)  transformation~\cite{jordan_uber_1928}.
Many efforts are directed towards lowering Pauli weight, that is the number of qubits that the encoded Fermionic operators act on, from the $\mathcal{O}(N)$ scaling of JW to more favourable $\mathcal{O}(\log N)$ scaling of Bravyi-Kitaev \cite{bravyi2002Fermionic} where $N$ is the number of modes simulated.  
Certain mappings have succeeded in reducing the number of qubits from the $N$-qubits required to simulate $N$-modes usually \cite{Bravyi2017}. 
A substantial body of work has concentrated on reducing both these Pauli and qubit requirements in lattice models \cite{Chiew2022, Setia2018, Steudtner2018, Whitfield, Chien2020, Steudtner2019}. 
Others have optimized measurement costs by introducing mappings with provably optimal Pauli weight \cite{jiang2020optimal}.
Recently, the connection between ternary trees and mappings has been explored \cite{vlasov2019clifford}. 

Additional work involves the study of custom encodings to reduce circuit overhead in the context of VQE, as highlighted in \cite{wang2023optimized}. 
In \cite{ChienPHASECRAFT2022}, a general scheme that employs a brute force search over the space of encodings mapping from Majorana monomials to Pauli operators is explored. 
These mappings are also optimized for limited qubit connectivity settings, with resulting encodings providing fairly general optimality guarantees on solutions. 
However, due to the high computational cost of the brute force method, only small systems are feasible with a focus on symmetric lattice models.
In \cite{Chiew2022}, the enumeration scheme between Fermionic modes and qubit operators representing said modes is explored to minimize various simulation costs with the JW encoding. 

When simulating a Fermionic state we are free to choose the mapping of the Fermionic-based operations comprising said state.
In our paper, we introduce ``treespilation", a technique that leverages this understanding and extends the Bonsai algorithm \cite{miller2022bonsai}.  
The algorithm optimises a mapping of $\ket{\psi _f}$ to prepare $\ket{\Psi_q}$ with a low CNOT count.
To illustrate this approach, we optimize the encodings for numerically produced ADAPT-VQE ansatz, using Fermionic and introduced Majoranic pools on setups with full and limited connectivity.
Comparing our approach to the non-Fermionic QEB and qubit pools, we observe that across the molecules considered, our method significantly outperforms these state-of-the-art approaches on setups with limited connectivity. 
When considering setups with full connectivity, on average, our approach shows improvement over using QEB and qubit pools, challenging the benefits of employing these widely used non-Fermionic operations in state preparation.
Of the methods considered, we find the Majoranic pool combined with treespilation to be by far the most CNOT-efficient pool for state preparation on limited connectivity hardware. 
Specifically, we observe that the limited connectivity overhead is eliminated, and in certain cases, the CNOT count is reduced compared to the full connectivity ansatz in the JW encoding. 
Figure \ref{fig:f2q opt procedure} illustrates the method.

\section{Preliminaries}
\subsection{Fermionic systems \label{sec:Fermionic systems}}
An $N$-mode Fermionic system in second quantization can be described in terms of $N$ creation operators $\{ a_i^\dagger \}_{i = 0}^{N-1}$ and annihilation operators $\{ a_i \}_{i = 0}^{N-1}$ that satisfy the canonical anticommutation relations:
\begin{gather}
    \label{eq:Fermionic_anticommutation}
    \{ a_i, a_j \} = \{ a_i^\dagger, a_j^\dagger \} = 0, \\ \{ a_i^\dagger, a_j \} = \delta_{ij} \mathbbm{1}.
\end{gather}
Mathematically, the $N$-mode Fermionic system is equivalent to the $N$-dimensional Fock space $\mathcal{F}(\mathbb{C}^N)$, a $2^N$-dimensional Hilbert space spanned by the so-called Fock basis. The operators defined above allow us to define the basis as follows. First, the Fermionic vacuum $\ket{\text{vac}_\text{f}}$ is defined to be the unique vector such that $a_j \ket{\text{vac}_\text{f}} = 0$ for all $j = 0, \dots, N-1$. The remaining basis elements can be constructed by considering all possible combinations of occupation numbers $n_j \in \{0, 1\}$:
\begin{align}
    \ket{n_0 n_1 \dots n_{N{-}1}} \coloneqq \prod_{j=0}^{N-1} (a_j^\dagger)^{n_j} \ket{\text{vac}_\text{f}}.
\end{align}

Creation and annihilation operators are not the only operators that can define the Fermionic space. It is also common to define an equivalent set of so-called Majorana operators $\{ m_k \}_{k = 0}^{2N-1}$ as
\begin{gather}
    \label{eq:Fermion_to_Majorana}
    m_{2 j} \coloneqq a_{j}^\dagger + a_{j},\\ m_{2 j+1} \coloneqq  \mathrm{i} (a_{j}^\dagger - a_{j}).
\end{gather}
Such operators obey many useful properties, such as being unitary and self-adjoint. Additionally, one can show that they obey the anticommutation relation $\{m_{i},m_{j}\} = 2 \delta_{ij} \mathbbm{1}$.

The above ways of defining Fermionic systems allow us to provide two equivalent forms of an $N$-mode second-quantized Fermionic Hamiltonian:
\begin{equation}
\begin{split}
    \mathcal{H}_f &= \sum_{ij}^{N-1} h_{ij} a_i^{\dagger} {a_j} + \sum_{ijkl}^{N-1} h_{ijkl} a_i^{\dagger} a_j^{\dagger} {a_k} {a_l} \\
    &=\sum_{ij}^{2N-1}\mathrm{i} c_{ij} {m_i} {m_j} + \sum_{ijkl}^{2N-1} c_{ijkl} \, {m_i} {m_j} {m_k} {m_l}.
\end{split}
\end{equation}
for coefficients $h_{ij}$, $h_{ijkl}$,  $c_{ij}$ and $c_{ijkl}$.
The equivalence between these forms comes directly from the linear dependency presented in Eq.~\eqref{eq:Fermion_to_Majorana}. 

An important operation we consider in this paper is the Majorana braiding transformation $U_{jk}$ \cite{majoranas}. 
This unitary swaps the roles of the $k$'th and $j$'th Majorana modes (up to a sign), leaving other Majoranas unchanged.
From consideration of the Fermionic parity \cite{Whymajoranasarecool}, the unitary can be expressed as the Clifford operator:
\begin{equation}
    \label{eq: braiding operator}
    U_{jk} = \exp{\frac{\pi}{4}{m}_j{m}_k} = \frac{1}{\sqrt{2}}\left(  \mathbbm{1} + {m}_j {m}_k \right),
\end{equation}
with the Majoranas transforming as:
\begin{equation}
    \begin{aligned}
        & {m}_j \rightarrow U_{jk} {m}_j U_{jk}^\dagger = -m_k\\
        & {m}_k \rightarrow U_{jk} {m}_k  U_{jk}^\dagger =m_j.
    \end{aligned}
\end{equation}
Considering $U_{2j,2j+1}$, we see it exchanges the role of Majoranas within Fermionic mode-$j$, that is ${m}_{2j} \rightarrow - {m}_{2j+1}$ and ${m}_{2j+1} \rightarrow {m}_{2j}$.
with creation and annihilation operators transforming as,
\begin{equation}
\label{eq: ferm braid}
    \begin{aligned}
        & a_j \rightarrow -\mathrm{i}a_j, \\
        & a_j^\dagger \rightarrow \mathrm{i}a_j^\dagger.
    \end{aligned}
\end{equation}
Under this particular transformation, Fock basis states are mapped to the same state with a phase shift, i.e., $\ket{n_0 n_1 \dots n_{N-1}}\rightarrow \mathrm{i}^{ n_j}\ket{n_0 n_1 \dots n_{N-1}}$. 
Importantly, the vacuum state is invariant for this transformation. 
For arbitrary exchanges, this is not the case, as we see when we consider swapping $m_{1}$ and $m_{2}$ that partially constitute modes $0$ and $1$. 
In this case, the vacuum state $\ket{00}$ is transformed into a nontrivial linear combination of the original Fock states
\begin{equation}
    \ket{00} \rightarrow U_{12}\ket{00}=\frac{1}{\sqrt{2}}(\ket{00}-\mathrm{i}\ket{11}),
\end{equation}

These features, wherein the Fock basis states are mapped to basis states and vacuum state preservation, hold significance for subsequent sections in which our objective is to perform Fermion-to-qubit transformations that maintain these features. Specifically, we aim to map Fock product states to qubit states while encoding the Fermionic vacuum state as the all-zero qubit state.

\subsection{Fermion-to-Qubit mappings}
The Fermionic Fock space, $\mathcal{F}\left(\mathbb{C}^N\right)$, and the Hilbert space of $N$ qubits $\left(\mathbb{C}^2\right)^{\otimes N}$ are both $2^N$-dimensional Hilbert spaces; thus, we can unitarily map between them. 
A natural unitary mapping is to map Fock basis states $\mathcal{F}\left(\mathbb{C}^N\right)$ to computational basis states of the qubits such that the occupation number of the $j$'th Fermionic mode matches with the state of the $j$'th qubit \cite{jordan_uber_1928}:
\begin{equation}
    \mathcal{F}\left(\mathbb{C}^N\right) \ni \left|n_0n_1 \ldots n_{N-1}\right\rangle \mapsto \bigotimes_{i=0}^{N-1}\left|n_i\right\rangle \in \left(\mathbb{C}^2\right)^{\otimes N}.
\end{equation}
This mapping is known as the Jordan-Wigner (JW) transformation, and under it, the basis of Majoranas is mapped to qubit space as:
\begin{equation}
    \label{eq:jw_majoranas}
    m_{2 j} \mapsto X_j \prod_{k=0}^{j-1} Z_k \quad \text{and} \quad m_{2 j+ 1} \mapsto Y_j \prod_{k=0}^{j-1} Z_k,
\end{equation}
for $j=0,1, \ldots N-1$. 
Here and in the rest of the paper, we denote $P_j$ with $P \in \{X, Y, Z\}$ for an operator that acts as the Pauli operator $P$ on the $j$'th qubit and as the identity on all others. 

The JW mapping is part of the class of the so-called \textit{Majorana string} Fermion-to-qubit mappings, as it associates a Majorana operator with a single Pauli string $S_i$ while preserving the commutation relations:
\begin{equation}
\begin{aligned} 
    & \{{m}_i, {m}_j\}  = 2\delta_{ij} \mathbbm{1} \rightarrow \{S_{i}, S_{j}\}  = 2\delta_{ij} \mathbbm{ 1}
\end{aligned}
\end{equation}
for $i,j \in \{0,\ldots,2N-1\}$.
This identification is implicitly used in other canonical mappings \cite{bravyi2002Fermionic, jiang2020optimal, jordan_uber_1928, parity}, and we refer to the associated Pauli operators $S_i$ as \textit{Majorana strings}. 
To complete the mapping, these Majorana strings are paired into qubit mode operators $A_i$ and $A_i^\dagger$. 
Then, the qubit vacuum state $\ket{\text{vac}}_\text{q}$ is found by solving $A_j\ket{\text{vac}}_\text{q}$ for $j \in 1 \dots N$. Note that many mappings exist outside this class \cite{setia_superfast_2019,bravyi2002Fermionic, kirby_second-quantized_2022}; however, this class proves particularly useful for quantum chemistry.

\subsection{Ternary tree based Fermion-to-qubit mappings}
We now present a useful class of Majorana-string mappings introduced in \cite{miller2022bonsai}, the product-preserving ternary-tree (PPTT) based Fermion-to-qubit mappings. 
These mappings possess the crucial property of product preservation, transforming Fock basis product states into qubit computational basis product states, i.e., $\ket{n_0 n_1 \dots n_{N-1}} \rightarrow \ket{x_0 x_1 \dots x_{N-1}}$ for $x_i \in \{0, 1\}$. 
This ensures the separability in qubit space for states such as the Hartree-Fock.
Another notable feature is that the Fermionic vacuum is transformed to the all-zero qubit state, i.e. $\ket{\text{vac}}_\text{f}\rightarrow \ket{0\dots0}$. 
This is particularly significant, as these states are the initial starting points for numerous quantum chemical algorithms, including UCCSD and ADAPT-VQE, and can thus be prepared without entangling gates.
Additionally, the authors establish a connection between the tree structure of the mapping and the encoding of Fermionic mode occupancy information in qubits. 
The methodology encompasses well-known mappings such as the Jordan-Wigner~\cite{jordan_uber_1928}, Bravyi-Kitaev \cite{bravyi2002Fermionic}, Ternary Tree \cite{jiang2020optimal}, and Parity \cite{Bravyi2017} encodings. 
Furthermore, they introduce the Bonsai algorithm, which leverages the flexibility of this methodology to design mappings that reduce SWAP gate requirements by aligning the generating ternary tree of the Fermion-to-qubit mapping with the limited connectivity of the quantum device.

\begin{figure}[t]
   \centering
   \includegraphics[width=0.95\columnwidth]{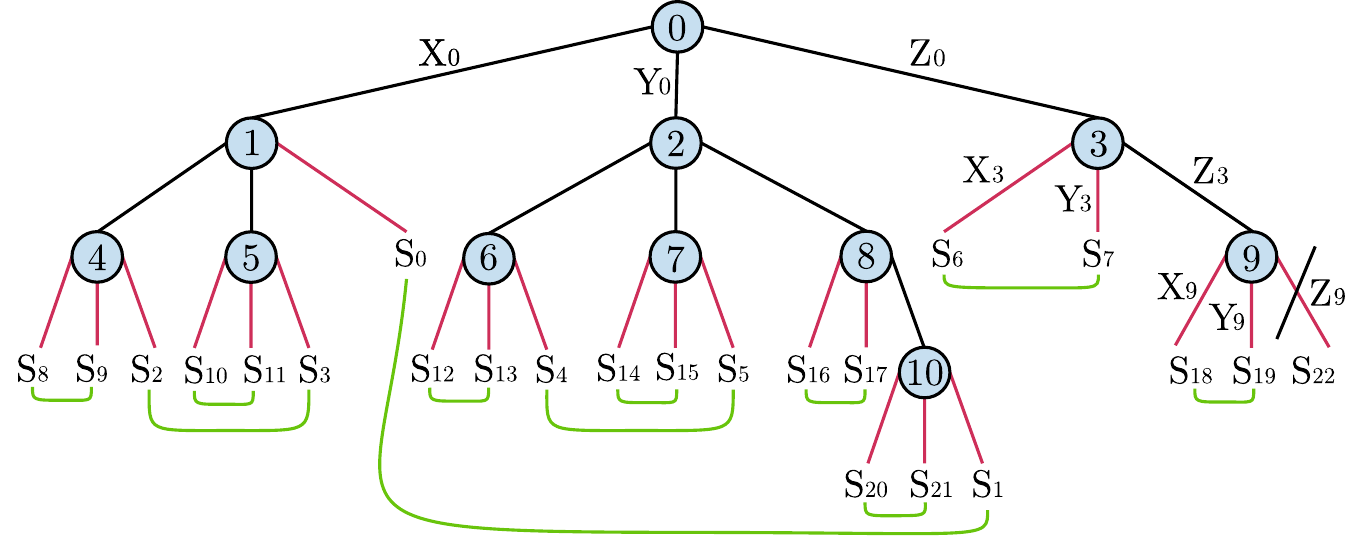}
   \caption{Example of a mapping derived from a ternary tree. 
   The enumeration of the vertices corresponds to the qubits while the black lines represent edges connecting two qubits and the red lines represent legs which we label with a Majorana string $S_u$.
   We represent the Pauli labelling on the links (both edges and legs) based on their position: the leftmost, central, and rightmost links below node-$u$ are labelled with $X_u$, $Y_u$ and $Z_u$ respectively.
   This labelling is explicit along the rightmost path (0-3-9).
   Each leg in the tree is associated with a Pauli string by following the path from the root node (node-0, in this case) to the leg. 
   The strings are generated as follows: each time a link with label $P$ stemming downward from a qubit-$u$ is crossed, the Pauli operator $P$ acting on qubit-$u$ is added to the string. 
   The resulting string acts trivially on all qubits not visited along the path, while differs by only a single Pauli from any other string. 
   To define the mapping, we remove the rightmost all-$Z$ operator ($S_{22}=Z_0Z_3Z_9$ in this case), and the remaining strings are paired into qubit modes according to the pairing algorithm outlined in \cite{miller2022bonsai} and in doing so we associate the Pauli strings $S_u$ with Majorana operators $m_i$. 
   Pairings are represented by the green lines in the figure. 
   This process ensures the separability of the Fermionic product states in qubit space and maps the Fermionic vacuum to the zero qubit state.
   For example, the leg $S_{18}$ corresponds to the string $S_{18} = Z_0Z_3X_9$, while $S_{19} = Z_0Z_3Y_9$.
   These strings are paired to represent the $i$th-mode, $a_i=\rightarrow\tfrac{1}{2}(X_9 + iY_9)Z_0Z_3$. 
   For clarity purposes, we omit the mode index and braiding flag from the node label.}
   \label{fig:example_tree}

\end{figure}
A PPTT mapping is uniquely defined by a labelled ordered TT. 
The ordered TT is a directed graph and a tree with $N$ nodes such that there is a unique node with indegree 0, and all the nodes point to at most 3 other nodes.
These nodes are usually labelled as left, middle, and right child, but for convenience, we will label them with $X$, $Y$, and $Z$, respectively. 

Given such a labelled ordered Ternary Tree (TT), one can uniquely construct a PPTT Fermion-to-qubit mapping as follows.
First, following the procedure presented in \cite{miller2022bonsai}, one can generate a basis of pairs of anticommuting Pauli strings from the tree that are later connected to particular Majoranas to complete the mapping.
The procedure starts by adding ``legs" to a ternary tree, which are labels associated with each vertex so that the total number of outward edges plus legs is equal to three for each node. 
Then, Pauli-$X$, Pauli-$Y$, or Pauli-$Z$ labels are assigned to each of the legs of the tree, so that each node has three edges or legs stemming from it with each of these labels. 
For such a tree with $N$-nodes, there are $(2N+1)$-paths from the root to the legs. 
Pauli operators are associated with each of these paths by taking the tensor product of each Pauli label along the path where the Pauli acts on the qubit-$u$.
All the Paulis generated this way pairwise anticommute, as they exhibit nontrivial differences at a single qubit, that is, there is one qubit at which they simultaneously have different a Pauli assigned such that neither is identity. 
This procedure generates $(2N+1)$-anticommuting Pauli strings that are to be associated with $2N$-Majorana operators.
This is achieved by removing the Pauli string that consists of only $Z$-Paulis, leaving us with $2N$-strings. 
Then we apply the pairing algorithm outlined in \cite{miller2022bonsai} to form $N$-Pauli operator pairs called \textit{qubit modes}. 
In doing so, we connect each string with a Majorana.
Each qubit mode is associated with the qubit in the TT at which the corresponding Majorana string pair exhibits a nontrivial difference.

A bijection is then established between the $N$-qubit modes and the $N$-Fermionic modes of our system. 
Any association ensures the PP property, as proven in \cite{miller2022bonsai}. 
Moreover, as demonstrated in Sec.~\ref{sec:Fermionic systems}, we can interchange the roles of our Majoranas within mode pairs while maintaining this property.
It is important to note that the Fermionic labelling of the Majoranas and modes involve Fermionic operations that do not alter the Pauli structure of the basis of strings generated by this construction.

Thus, we redefine a PPTT mapping by a set of $N$-nodes, each labelled with a triple $(j, u, b) \in [N] \times [Q] \times \{+,-\}$, where $j$ is the index for the Fermionic mode, $u$ is the qubit index of the $Q$-qubit machine, and $b$ specifies the Majorana order in the mode pair. 
Note that each value of $j$ and $u$ can appear only once in the tree. 
For each node with a label $(j,u,b)$, mode-$j$ is assigned to the assigned pair of Pauli strings that exhibit nontrivial difference at qubit-$u$. 

It can be shown that the Majorana operators assigned to a particular node for $b=$`+' correspond to
\begin{gather} \label{eq:PPTT-majoranas}
        m_{2j} \rightarrow S_{s_x^{(u)}} = X_u \prod_{k \in \mathcal{Z}_{x}^{(u)}} Z_k G_{u}, \\
        m_{2j+1} \rightarrow S_{s_y^{(u)}} = Y_u \prod_{k \in \mathcal{Z}_{y}^{(u)}} Z_k G_{u}.
\end{gather}
When we exchange the Majoranas, i.e. $b=$`$-$', we get
\begin{gather} \label{eq:PPTT-majoranas-braided}
        m_{2j} \rightarrow -S_{s_y^{(u)}} = -Y_u \prod_{k \in \mathcal{Z}_{y}^{(u)}} Z_k G_{u}, \\
        m_{2j+1} \rightarrow S_{s_x^{(u)}} = X_u \prod_{k \in \mathcal{Z}_{x}^{(u)}} Z_k G_{u}.
\end{gather}
where $\mathcal{Z}_{x}^{(u)}$ and $\mathcal{Z}_{y}^{(u)}$ are sets of qubits that $S_{s_x^{(u)}}$ and $S_{s_y^{(u)}}$ act non-trivially on below qubit $u$ in the tree, and $G_u$ is a common Pauli string. 
This equation can be graphically understood as $G_u$ being the common path of $S_{s_x^{(u)}}$ and $S_{s_y^{(u)}}$ from the root to qubit-$u$ and sets $\mathcal{Z}_{x}^{(u)}$ and $\mathcal{Z}_{y}^{(u)}$ are qubits along the $Z$-paths bifurcating from the $X$- and $Y$-legs of qubit-$u$. 
Note that $ \mathcal{Z}_{x}^{(u)} \cap \mathcal{Z}_{y}^{(u)} =\emptyset$ for any $x,y$. 

The mapped creation and annihilation mode operators for $c=$`$+$' are thus,
 \begin{equation}
     \label{eq:general_form_tree_op}
     \begin{aligned}
             &a_j \mapsto \frac{1}{2} \left( X_{u} \prod_{k \in \mathcal{Z}_{x}^{(u)}} Z_k + i Y_u \prod_{k \in \mathcal{Z}_{y}^{(u)}} Z_k \right) G_u,\\
             &a_j^{\dagger} \mapsto \frac{1}{2} \left( X_{u} \prod_{k \in \mathcal{Z}_{x}^{(u)}} Z_k - i Y_u \prod_{k \in \mathcal{Z}_{y}^{(u)}} Z_k \right) G_u.
     \end{aligned}
\end{equation}
When the roles in the pair are switched, i.e. $c=$`$-$', a complex phase term emerges as in Eq.~\eqref{eq: ferm braid}.

The PP property is a crucial attribute of the mappings, and the pairing algorithm outlined in \cite{miller2022bonsai} is proven to ensure the PP property.
To understand the significance of this property, consider the vacuum state in a PPTT mapping $F_{\text{PP}}$, i.e., $\ket{\text{vac}}\rightarrow \ket{0}$.
In $F_{\text{PP}}$, we associate Majoranas with Pauli strings, i.e.
$(m_0, m_1, \dots , m_{2N-1}) \leftrightarrow (S_0, S_{1}, \dots , S_{2N-1})$ where this notation means we associate $m_0 \leftrightarrow S_0, m_1 \leftrightarrow S_1$ and so on. 
Now envision a mapping $F_{\text{NPP}}$ that differs from $F_{\text{PP}}$ by the assignment of Majoranas to Pauli operators. 
This deviation is captured by permutation $p$ from the initial order in $F_{\text{PP}}$, i.e.  $(m_0, m_1, \dots , m_{2N-1}) \leftrightarrow (S_{p(0)}, S_{p(1)}, \dots , S_{p(2N-1)})$.
The vacuum state in the new mapping can be related to the initial state by $|\tilde{0}\rangle = T \ket{0}$, for some unitary $T$.
To identify this operator, we can break down the permutation to pairwise transpositions, $p = (t_{0}, t_{1})(t_2, t_3)\ldots (t_{k-1}, t_{k})$.
Each transposition can be interpreted as a pairwise exchange of the Majorana strings within qubit modes, where the Fermionic representation is given by Eq.~\eqref{eq: braiding operator}.
Thus, we can express the vacuum state in $F_{\text{NPP}}$ as,
\begin{equation}
    |\tilde{0}\rangle = e^{\tfrac{\pi}{4} S_{t_{k-1}}S_{t_{k}}} \ldots e^{\tfrac{\pi}{4} S_{t_{2}}S_{t_{3}}} e^{\tfrac{\pi}{4} S_{t_{0}}S_{t_{1}}} \ket{0}.
\end{equation}
Therefore, to prepare the vacuum trial state in a non-PP TT mapping, one must apply several of these pairwise braiding unitaries to the original PPTT vacuum state.
In practice, there may be prohibitively many transpositions, and the product $S_i S_j$ can be highly nonlocal, resulting in a large gate cost to merely implement the vacuum state. 
Likewise, these exchanges must be implemented to prepare the Hartree-Fock state defined as the Fermionic product state, $\prod_{k\in O} a^\dagger_k \ket{\text{vac}}$, where $O$ is the set of occupied spin-orbitals.
As in this paper, we seek to reduce the number of CNOT gates, we strictly consider the space of PPTT mappings.

\subsection{State preparation}
Quantum chemistry calculations using VQE-ansatz-based approaches involve the preparation of a parameterised quantum circuit on the qubits of a quantum device, followed by appropriate measurements to determine the desired chemical properties. 
On a quantum computer, we are usually restricted to applying unitary operations to construct this ansatz. 
In quantum chemistry, these unitaries can be derived from chemical principles such as the single- and double-excitation operators in Unitary Coupled-Cluster techniques.
These operations in particular are important as they can construct an ansatz to approximate an electronic wavefunction to arbitrary accuracy \cite{PhysRevLett.84.2108singlesanddoublesareenough,whitfield2011simulation-singlesanddoublesareenough}.
We can express such a Fermionic state, denoted $\ket{\psi_f}$, as a product of $r$-unitaries applied to a reference state, $\ket{\psi_{\rm ref}}$, as:
\begin{equation}
    \label{eq:fermi circ}
    \ket{\psi_f} =\prod_k{u}_k(\theta_k)\ket{\psi_{\rm ref}} = \prod_k e^{ \theta_k {\tau}_k} \ket{\psi_{\rm ref}},
\end{equation}
where ${u}_k(\theta) = \exp( \theta {\tau}_k)$ are unitaries parameterised by $\theta_k$ generated by the Fermionic generators ${\tau}_k$.
This Fermionic circuit is subsequently mapped to qubit space using a Fermion-to-qubit mapping,
\begin{equation}
    \label{eq:fermi circ mapped}
    \ket{\psi_f} \rightarrow \ket{\Psi_q} = \prod_k e^{ \theta_k {T}_k} \ket{\Psi_{\rm ref}},
\end{equation}
where $T_k$ and $\ket{\Psi_{\rm ref}}$ are the qubit generators and reference state. 
The parameters $\theta_k$ are then optimized on the quantum computer with a classical algorithm.

The ADAPT-VQE approach \cite{adapt-vqe-orig-paper} has emerged as a promising avenue for quantum-based chemical state preparation. 
This technique involves iteratively applying parameterized qubit unitaries from a predefined set of operators, referred to as a ``pool", to a reference state while minimizing the energy following the variational principle. 
The choice of the pool significantly impacts both the convergence of the ADAPT algorithm and the gates in the quantum circuit representation of the ansatz.

A common choice is the \textit{Fermionic pool} \cite{adapt-vqe-orig-paper}, comprising single- and double-excitation operations that preserve the spin and particle number of the resulting state. The single and double excitation operators generate the pool:
\begin{gather}
\label{eq: fermi pool}
{\tau}_{j}^{i} = {a}_i^{\dagger} {a}_j - {a}_j^{\dagger} {a}_i \\
{\tau}_{ij}^{kl} = {a}_i^{\dagger} {a}_j^{\dagger} {a}_k {a}_l - {a}_k^{\dagger} {a}_l^{\dagger} {a}_i {a}_j,
\end{gather}
For $j,k \in \{0,\dots, N-1\}$.
We can express these excitation elements in terms of a linear combination of products of underlying Majorana Fermions using Eq.~\eqref{eq:Fermion_to_Majorana}. 
We define the \textit{Majoranic} pool by taking each element of this linear combination and adding it to the pool separately for each element in the Fermionic pool:
\begin{gather}
\label{eq: maj pool}
{\tau}^{v}_{u} = {m_u}{m_v} \\
{\tau}^{uv}_{rs} = \mathrm i {m_u}{m_v}{m_r}{m_s}
\end{gather}
for $u,v,r,s \in \{0,\dots, 2N-1\}$. 

Both of the aforementioned pools possess a proper Fermionic representation, thus ansatz constructed from them have a representation in Fermionic space given by Eq.~\eqref{eq:fermi circ}. 
However, accounting for Fermionic anticommutation when mapping to qubits generally leads to highly nonlocal operators in qubit space. 
The \textit{Quantum-excitation-based} (QEB) \textit{pool} \cite{yordanov2021qubit} aims to rectify this nonlocality by not implementing the exact commutation relations of the operators. 
It is derived by mapping the Fermionic-spin pool using the JW transformation and removing the trailing $Z$-strings, yielding qubit operators:
\begin{gather}
 T_{j}^i =  \frac{\mathrm i}{2}\left(X_i Y_j-Y_i X_j\right), \\ %
    \begin{aligned}
    T^{ij}_{kl} &=  \frac{\mathrm i}{8}(X_i Y_j X_k X_l + Y_i X_j X_k X_l + Y_i Y_j Y_k X_l \\
    & \phantom{ \frac{\mathrm i}{8}(}+ Y_i Y_j X_k Y_l - X_i X_j Y_k X_l - X_i X_j X_k Y_l \\
    & \phantom{\frac{\mathrm i}{8}(}- Y_i X_j Y_k Y_l - X_i Y_j Y_k Y_l).
    \end{aligned}
\end{gather}
The elimination of the parity-checking $Z$-strings is motivated by the existence of efficient circuit representations that can be achieved using a constant number of CNOT gates, resulting in at most 4-local gates. 
Due to the removal of the parity-checking strings, the Fermionic interpretation for this pool is unclear, thus it has been coined ``pseudo-Fermionic". 
Similar to how we did with the Majoranic pool, the QEB pool can be broken down further into the so-called \textit{qubit-pool} \cite{tang2021qubit}, where in this case the unitaries generated by single strings of Pauli operators of weight 2 or 4.
The QEB and qubit-pool are defined only for the Jordan-Wigner mapping, unlike the Fermionic-based pools.

\section{Treespilation: optimizing the Fermion to qubit mapping}
A Fermion-to-qubit mapping is chosen to encode Fermionic operations and the Hamiltonian in preparing ansatz states. 
The choice of mapping is not unique, and it can significantly impact the efficiency of the resulting circuit. 
In the case of the Fermionic and Majoranic pools, we are entirely free to choose a mapping. 
In a typical workflow, however, a single mapping is selected, and the resulting circuit is optimized with a \textit{transpiler}.

We can in general, exploit this freedom of mapping to optimise a given state. 
To do so, one would need the following elements:
\begin{enumerate}
    \item \textbf{A target Fermionic state, denoted} \(\ket{\psi_f}\): 
    The quantum state must be in a form that allows one to effectively find the corresponding qubit representation.
    Quantum states in Eq.~\eqref{eq:fermi circ} satisfy this requirement.
    
    \item \textbf{A cost function} \(\mathcal C\): For a given Fermionic state \(\ket{\psi_f}\) and mapping \(F\), it computes the quality of the qubit state \(\mathcal C(\ket{\psi_f}, F)\).  
    As this paper focuses on minimizing the number of CNOTs for a transpiled circuit representing the state \(F(\ket{\psi_f})\), a natural choice is to use this number as the cost function. 
    However, as we will show later, other cost functions can be used.
    
    \item \textbf{A set of rules that allow us to transform a given F2Q mapping} \(F\) \textbf{into another} \(F'\): This is essential for optimization algorithms. 
    Various methods, such as simulated annealing, may require rules for generating new mapping candidates based on optimization history.
    
    \item \textbf{An optimization algorithm}: Equipped with the cost function defined above, it searches for high-quality Fermion-to-qubit mappings. 
    Since the space of PPTT mappings is discrete and large, a natural choice is to use metaheuristic algorithms.
\end{enumerate}

In this section, we develop an approach to systematically explore various Fermion-to-qubit mappings to find efficient circuit representations of \(\ket{\psi}\). 
Specifically, we will start by proposing multiple cost functions considered in this paper. 
Then we will introduce the simulated annealing optimization algorithm considered in this paper and present transformations allowing us to explore the space of PPTT mappings. 
Finally, we will present the treespilation algorithm.

\subsection{Cost Functions \label{sec:cost functions}}
We now discuss the cost functions used to optimize the qubit representations, $\ket{\Psi_{q}}$, of $\ket{\psi_{f}}$. 
In present-day quantum devices, the number of CNOT gates is a crucial metric for analysing the feasibility of applying quantum circuits on quantum hardware. 
These gates typically take more time and introduce errors that are approximately ten times higher than those of single-qubit gates~\cite{gateerrors}. 
Consequently, the cost functions employed in this context will prioritize minimizing the CNOT count. 
To effectively optimize, an appropriate cost function must be accurate and fast.

\paragraph{Transpiler cost} 
The ideal cost function involves processing a qubit state $\ket{\psi_q}$ using a transpiler, like those in \textit{Qiskit}~\cite{aleksandrowicz_qiskit_2019} or \textit{TKET}~\cite{Sivarajah_2021}, to calculate metrics such as depth or CNOT count. 
Although this approach offers high accuracy and optimization specific to the transpilation scheme, it might be too slow in practice.

\paragraph{Pauli string cost} 
Alternatively, one can consider a simpler cost based on the Pauli representation of the generators $T_i$ that make up $\ket{\Psi_q}$ as defined in Eq.~\eqref{eq:fermi circ mapped}. 
On a fully connected architecture, implementing a Pauli string represented as $P$ requires $2(k-1)$ CNOT gates, where $k$ is the Pauli weight, i.e. the number of non-identity terms in the string. 
However, on a limited connectivity architecture, implementing the same string necessitates $2(2n - k - 1)$ CNOT gates, where $n$ is the number of nodes in the Steiner tree spanning qubits acted on by Paulis in the string~\cite{Vandaele2022} (i.e., a minimal subtree of the hardware architecture containing the qubits acted on by the Paulis in this string). 
For our purposes, this cost estimates the number of CNOTs by iterating over all the generators $T_i$ that appear in the qubit state $\ket{\Psi_q}$ as defined in Eq.~\eqref{eq:fermi circ mapped}, and then estimating the sum of the number of CNOTs required for all the Paulis that appear in the following generator. 
The problem of finding an optimal Steiner tree is NP-hard in general~\cite{Karp1972}.
However, provided the underlying mapping tree is a subtree of the hardware architecture, i.e. it is a mapping defined by the Bonsai algorithm in \cite{miller2022bonsai}, the strings resulting from single and double excitation operations can be implemented optimally in polynomial time, as we prove in Appendix~\ref{app:steiner-tree}.
This cost function provides a robust measure for Majoranic pools but overlooks CNOT cancellations between adjacent gates. 
In the Fermionic pool, where substantial CNOT cancellations occur within generators, this simplistic approach less accurately tracks CNOT costs. 
Improvement may be achieved by considering the collective structure of Pauli strings within these generators.

\subsection{Updating a tree-based mapping \label{sec:altering}}
When defining a PPTT mapping, we can adjust four degrees of freedom:
\begin{enumerate}
    \item Structure of the ternary tree. \label{item 1}
    \item \label{item 2} Choice of the root node. 
    \item \label{item 4} Association of Fermionic modes to qubit operators. 
    \item Ordering of Majoranas in the Fermionic mode pairs. \label{item 5}
\end{enumerate}
The most basic transformation updating point \ref{item 1} involves moving a leaf to a leg, effectively converting that leg into an edge of the ternary tree.
In this way, we can change both the Pauli-labelling and the edge structure of the tree.
The second degree of freedom in point \ref{item 2} is the ability to switch the designation of the root node.
Coupled with \ref{item 1}, we can deform between this class of mappings as we can from one ordered ternary tree to any other.

\begin{figure}[t]
   \centering
   \includegraphics[width=.95\columnwidth]{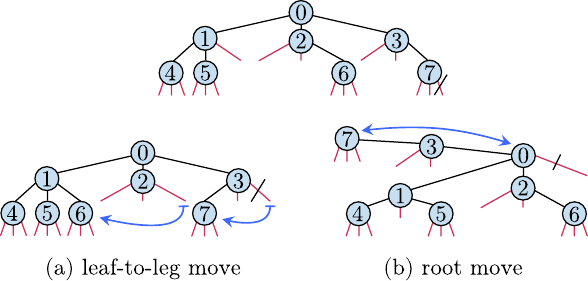}
   \caption{The basic tree transformations allow us to transition between different tree-based mappings. 
   (a) depicts the deformation process of moving external leaves to free legs. 
   This transformation may modify the underlying edges of the ternary tree, as demonstrated by the removal of edge $(2,6)$ and the addition of edge $(1,6)$. Furthermore, the labelling of the edges can be adjusted, as shown by the change from a $Z$ label to a $Y$ label for edge $(3,7)$. 
   (b) illustrates the ability to change the root node.}
   \label{fig: tree deformations}
\end{figure}

Item \ref{item 4} relates to our freedom to choose a bijection between the Fermionic modes and the pairs of qubit Majorana strings representing them.
Furthermore, \ref{item 5} exploits the freedom to switch or braid the Majoranas within the pairs defined in Eq.~\eqref{eq:Fermion_to_Majorana}.
Unlike the tree transformations, these updates do not change the Paulis in the Majorana strings that express the Fermionic system, they only change how said strings are associated with Fermionic mode operators.

All these transformations result in mappings that fall within the vacuum preserving and PP categories.
Notably, the braiding in \ref{item 5} expands the space of possible PPTT mappings beyond the original presentation in \cite{miller2022bonsai}, where braiding was not considered. 
Figure \ref{fig: tree deformations} visually illustrates some of these degrees of freedom.

\subsection{Optimisation}
Using a cost function and transformation rules, we now employ an optimization procedure to find a mapping that minimises the cost of implementing $\ket{\psi_f}$. 
In this work, we adopt simulated annealing as the chosen scheme. 

Simulated annealing is a probabilistic optimization technique inspired by the annealing process in metallurgy. 
It begins with a solution and iteratively explores potential solutions by introducing random changes. Once a new candidate $x'$ is created from $x$, the algorithm evaluates the objective values $E(x)$ and $E(x')$ and decides whether a new candidate should be accepted based on the difference between these objective values. If candidate $x'$ has a smaller objective value, then it is always accepted. At the same time, if the objective value is larger, it is only accepted with a decreasing probability depending on the ``temperature" $t$ that changes during the optimization. 
The probability of accepting a worse solution at iteration \(k\) is $\exp{\left(E(x) - E(x')\right)/t_k}$,
where $t_k$ is the temperature at step $k$.
This probabilistic acceptance of worse solutions allows the algorithm to escape local optima. 
This method applies to a wide range of complex optimization problems, especially those with non-convex or high-dimensional solution spaces where finding optimal solutions is challenging.
Simulated annealing requires defining a way of applying ``local'' changes to a solution, and the aforementioned transformation rules offer exactly this.
We note that more advanced optimization algorithms like Tabu search may be used.
Bringing this together we introduce \textit{treespilation}, for optimisation of a Fermionic states' mapping.

\section{Methodology} \label{sec:methodology}
We now describe the methodologies used to produce the results presented in this paper.
On full-connectivity devices (FC), we initiate the annealing process with a randomly generated mapping tree. 
For limited-connectivity (LC), we start with a Bonsai transformation of the device \cite{miller2022bonsai} so the tree underlying the mapping is connected on the device, and from it, we can determine the mapping from virtual to physical qubits.

For optimization purposes, we use the following adapted mapping tree transformations:
\begin{enumerate}
    \item \emph{Leaf move}: For FC, choose a random terminal node-$v$ with three legs and attach it to the free leg of another node-$w$ that is neither $v$ nor its parent. 
    For LC constraints, two choices are available: If starting from a Bonsai transformation, i.e., the mapping tree is connected on the device, the new qubit that node-$v$ assumes must be physically connected on the device to the qubit that parent node-$w$ represents. 
    We refer to this restriction as \textit{Connectivity Preserving} (CP), as the tree remains a connected tree on the underlying connectivity. 
    Alternatively, one can perform \textit{Non-Connectivity Preserving} (NCP) optimization, where this constraint is not applied, and the process proceeds as if it were FC. 
    The CP and NCP optimization strategies are demonstrated in Figure~\ref{fig:CP and NCP mappings}.
    
    \item \emph{root change}: 
    A node $v$ different than the root with out-degree at most 2 is chosen as a new root. 
    The path from root to $v$ is identified, and the tree is updated so child and parent designations are swapped along the path.
    
    \item \emph{Pauli shuffle}: 
    A random node with an out-degree of at least 1 is chosen and the Pauli operators associated with the links are changed.

    \item \emph{mode association swap}: 
    Two nodes with labels $(i,u,b)$ and $(i',u',b')$ are chosen at random and their labels are changed to $(i',u,b)$ and $(i,u',b')$ respectively.

    \item \emph{Majorana braiding change}: 
    A node with label $(i,u,b)$ is chosen at random, and the braiding $b$ is changed to the opposite one, i.e., `+' is changed to `$-$' and vice versa.
\end{enumerate}
While it is feasible to propose an alternative set of transformations, we demonstrate in Appendix~\ref{sec:reachability} that for heavy-hexagonal and 2D grid hardware graphs, it is possible to convert any given PPTT F2Q mapping into another using only connectivity-preserving leaf moves and the four other transformations. 
This proof establishes that these transformations provide sufficient control to derive a high-quality mapping. 
Additionally, one might anticipate a significant increase in possibilities with the growth of qubit numbers in the device. 
However, as detailed in Appendix~\ref{sec:number-of-f2q}, the number of PPTT F2Q mappings on bounded-degree hardware graphs is roughly equal to the total number of PPTT F2Q mappings. 
Furthermore, the dependency on the number of qubits can effectively become negligible.

We iteratively update the mapping with simulated annealing, choosing to propose one of the random updates above with equal likelihood.
For the Fermionic pool, we note the CNOT count of the generator is invariant to pairwise braiding so we do not use them here.
As mentioned, we consider the CP and NCP search settings of the algorithm.
Moreover, we also search in the restricted space of mappings generated by fixing the underlying mapping as JW and optimizing the assignment of qubit modes to Fermionic modes.
We name this setting Mode Shuffling (MS).

\begin{figure}[t]
    \centering
    \includegraphics[width=0.99\linewidth]{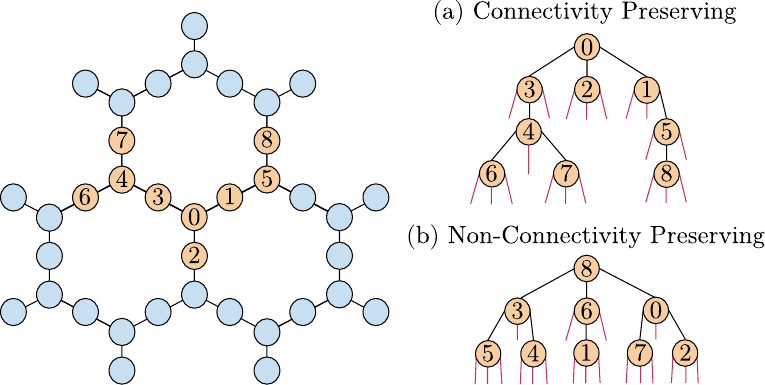}
    \caption{Illustration of Connectivity Preserving (CP) and Non-Connectivity Preserving (NCP) updates for a 9-mode PPTT mapping on a 37-qubit heavy-hexagon processor.
    The underlying mapping tree for CP, (a), is made of edges that are connected on the device targeted. For NCP, (b), the edges of the tree do not exist on the device}
    \label{fig:CP and NCP mappings}
\end{figure}

To benchmark our approach, we compare it with Fermionic and Majoranic pools mapped using the JW encoding with gate compilation and transpilation being pool and mapping dependant.
For the Fermionic pool, we utilize a circuit representation of excitation generators from \cite{Yordanov2020}, known for its CNOT efficiency in the JW encoding to compile the ansatz initially. 
For MS with this pool, we benefit from this representation as we use the JW encoding. 
After optimizing the mapping beyond simple mode association permutation, we employ the TKET compilation pass from \cite{Cowtan2020AGC} for effective gate compilation.

For the Majoranic pool on FC and LC, we initially compile generators using the standard CNOT staircase approach.
After treespilation on FC, we apply the same scheme.
On LC with CP optimization, we use a Steiner-tree compilation~\cite{Vandaele2022} with an optimal Steiner-tree generation (details in Appendix \ref{app:steiner-tree}). 
The assignment from logical to physical qubits, necessary for the latter approach, is provided by the underlying mapping tree.
For NCP, we compile strings using the standard CNOT staircase approach since the underlying tree is not connected on the device. 

QEB and qubit pools are analyzed using the representations in \cite{Yordanov2020} and the staircase approach, respectively. 
In all cases, after compiling pool elements in the ansatz, we use the \textit{Qiskit} transpiler at optimization level 3 for circuit optimization. 
For all cases on LC, bar CP, the transpiler is used to find an assignment between virtual and physical qubits.
Refer to Table \ref{tab:transpilers} in the appendix for a summary of compilation and transpilation passes used.

We assess our technique with both the Pauli and Transpiler cost functions for the Fermionic and Majoranic pools. 
We use the same passes as described previously for the transpiler cost function. 
To mitigate the stochastic nature of the Qiskit transpiler and the simulated annealing optimization scheme, was run five times and the minimal cost was selected. 
With each call of the transpiler cost, we run the transpiler a single time.

We evaluate our technique using both the Pauli and Transpiler cost functions for the Fermionic and Majoranic pools. 
The transpiler cost function employs the same passes as described earlier. 
To address the stochastic nature of the Qiskit transpiler and the simulated annealing optimization scheme, we run them five times and select the minimal cost for the final results. 
Each transpiler cost assessment involves a single transpiler run.

Our study focuses on the ground state of six molecules: LiH, N$_2$, BeH$_2$, H$_6$, C$_6$H$_8$, and C$_8$H$_{10}$. 
When mapped to qubits, these molecules require 12, 12, 14, 12, 12, and 16 qubits, respectively, for the active spaces and basis sets detailed in Table~\ref{tab:Molecular data}. 
Ansatz are prepared through classical numerical simulations using the ADAPT-VQE algorithm, with results within an error margin of $10^{-3}$ Hartree compared to the exact diagonalization result.

Plots illustrating the convergence of energy against CNOT gates are generated. 
After each excitation gate is added to the ansatz in the ADAPT procedure, it undergoes treespilation. 
We compare CNOT counts for full connectivity as well as IBM Eagle heavy-hexagon and grid-based Google Sycamore architectures. 
The connectivity graphs are displayed in the Appendix (see Section~\ref{fig:connectivity graphs}).

\section{Results}

\begin{table}[t]
    \begin{ruledtabular}
        \begin{tabular}{lcccccc}
        Molecule             & \multicolumn{2}{c}{Fermionic}      & \multicolumn{2}{c}{Majoranic}      & QEB      & Qubit        \\
                             & \textit{In.}     & \textit{Fin.}   & \textit{In.}     & \textit{Fin.}   &          &           \\\hline
        \multicolumn{7}{c}{$(a)$ \textbf{\textit{Full Connectivity}}}                                                         \\
        LiH (12)             &  158            & 98               & 154              & \textbf{40  }   & 97            & 74        \\
        N$_2$ (12)           &  452            & 326              & 398              & \textbf{194 }   & 572           & 434       \\
        BeH$_2$ (14)             &  580            & \textbf{356}     & 962              & 438             & 554           & 774       \\
        H$_6$ (12)           &  1660           & 1442             & 2398             & 1836            & \textbf{1400} & 1814      \\
        C$_6$H$_8$ (12)      &  362            & 262              & 274              & \textbf{122 }   & 216           & 156       \\
        C$_8$H$_{10}$ (16)   &  564            & 436              & 548              & 350             & 404           & \textbf{300}  
        \\[1ex]   \hline
        \multicolumn{7}{c}{$(b)$ \textbf{\textit{Google Sycamore}}}                                                            \\
        LiH (12)             &  325            & 148              & 198              & \textbf{72  }   & 184      & 120        \\
        N$_2$ (12)           &  1001           & 641              & 519              & \textbf{264 }   & 1236     & 890        \\
        BeH$_2$ (14)             &  1289           & 729              & 1438             & \textbf{612 }   & 1214     & 1646       \\
        H$_6$ (12)           &  3802           & 3359             & 3681             & \textbf{2216}   & 3117     & 3896       \\
        C$_6$H$_8$ (12)      &  804            & 523              & 356              & \textbf{150 }   & 415      & 278        \\
        C$_8$H$_{10}$ (16)   &  1227           & 914              & 750              & \textbf{420 }   & 870      & 599      
        \\[1ex] \hline
        \multicolumn{7}{c}{$(c)$ \textbf{\textit{IBM Eagle}}}                                                                   \\ 
        LiH (12)             &  463            & 224              & 271              & \textbf{102 }   & 256      & 202         \\
        N$_2$ (12)           &  1585           & 1006             & 665              & \textbf{350 }   & 1758     & 1444        \\
        BeH$_2$ (14)             &  2059           & 1013             & 1917             & \textbf{824 }   & 1696     & 2715         \\
        H$_6$ (12)           &  5783           & 4888             & 5261             & \textbf{2930}   & 4506     & 6284         \\
        C$_6$H$_8$ (12)      &  1276           & 770              & 490              & \textbf{202 }   & 598      & 445          \\
        C$_8$H$_{10}$ (16)   &  1898           & 1358             & 943              & \textbf{566}    & 1212     & 913          \\
        \end{tabular}
    \end{ruledtabular}
    \caption{CNOT counts for groundstate preparation circuits of various molecules obtained using the ADAPT-VQE algorithm with different choices of pools transpiled on both full and limited connectivity quantum devices. 
    The number of qubits used is indicated in brackets beside the molecule label. 
    Counts in bold correspond to the minimum of the rows.
    The initial (\textit{In.}) columns display counts of transpiled results using the JW encoding, while the final (\textit{Fin.}) column shows the results after treespilation and transpilation. 
    Treespilation is only applied to ansatz with clear Fermionic representations, and as such, it is not used with qubit or QEB pools. 
    For detailed information regarding molecular geometries, ADAPT-VQE convergence, device topologies, and transpilers used refer to the Appendix section~\ref{app:molecular data}. \label{tab:molecular results}}
\end{table}

\begin{figure*}[ht]
    \centering
    \includegraphics[width=0.95\linewidth]{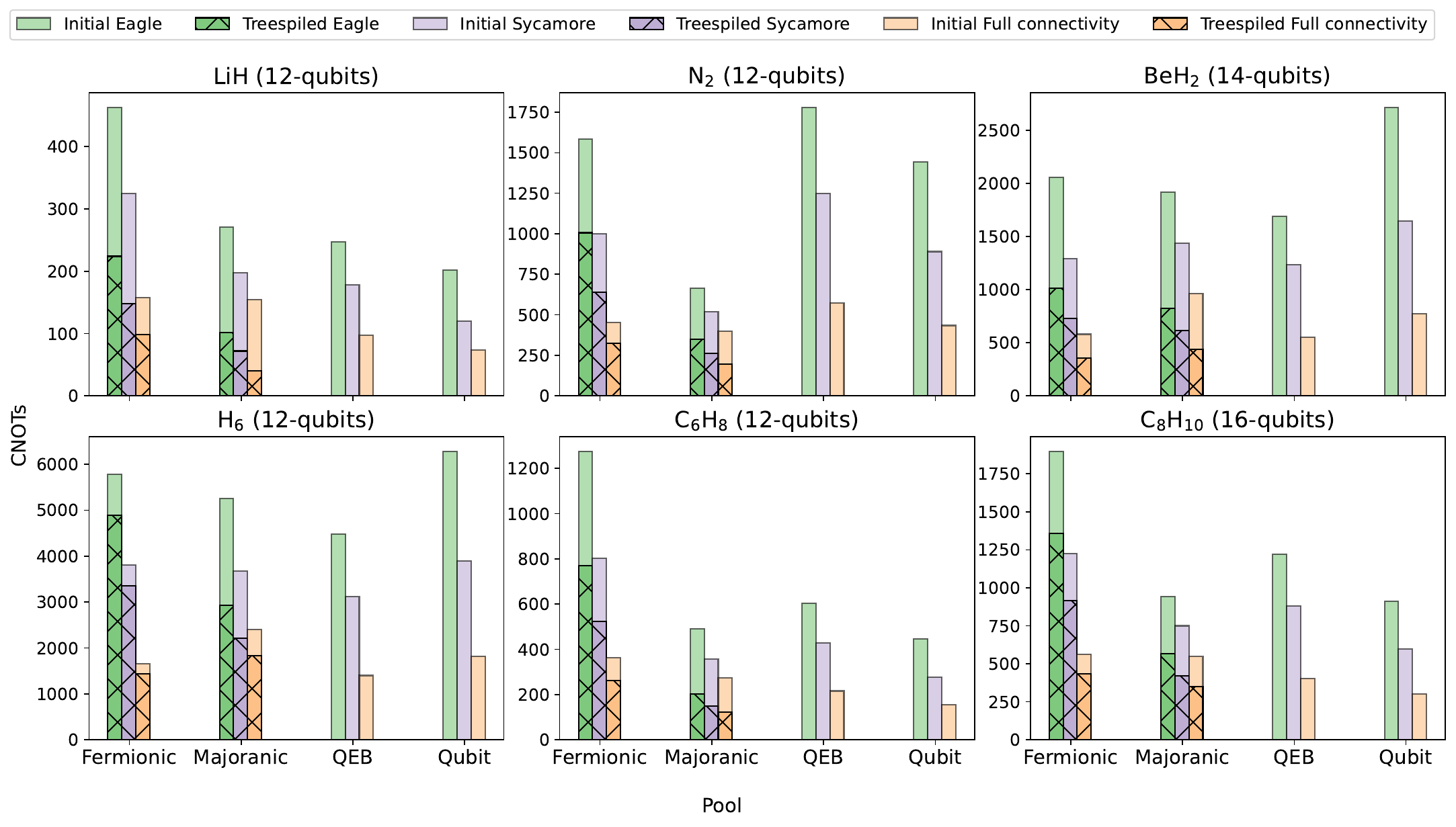}
    \caption{Visualisation of results tabulated in~\ref{tab:molecular results}. }
    \label{fig:molecular results}
\end{figure*}
From analysing the number of generators in the ADAPT simulations, in each case, we see the Fermionic and Majoranic pools converge to $10^{-3}$ precision in fewer or the same number of parameters compared to their non-Fermionic counterparts, the QEB and qubit pools respectively. 
This suggests that the Fermionic ground state can be more easily expressed using true Fermionic operations. 
Thus, we identify a potential tradeoff: using Fermionic operators allows for reaching a given precision in fewer iterations while using non-Fermionic counterparts may require more iterations but potentially fewer CNOTs.

In Tables~\ref{tab: FC results}, \ref{tab: Sycamore results}, and \ref{tab: Eagle results}, we present the Full Connectivity (FC) and Limited Connectivity (LC) results of our treespilation algorithm. 
For FC, we considered unconstrained and Mode Shuffling (MS) search settings with both Transpiler Cost (TC) and Pauli Cost (PC) functions. 
On LC, we explore MS, and Connectivity Preserving (CP) and Non-Connectivity Preserving (NCP) settings with both costs. 
Overall, we observe a significant reduction in the CNOT cost of implementing these states with treespilation.

Specifically, for the Majoranic pool on FC, the CNOT count reduced by an average of $49\%$, and with the Mode Shuffling (MS) search space, the reduction was $33\%$, with little difference between the Transpiler Cost (TC) and Pauli Cost (PC) functions. 
For LC, we see reductions of $51\%$ and $49\%$ for the Connectivity Preserving (CP) and Non-Connectivity Preserving (NCP) search spaces, respectively.
This indicates that the CP space, corresponding to Bonsai mappings with pairwise mode braidings, contains high-quality solutions. 
We also found that the PC cost effectively represents the CNOT cost and can replace the more time-consuming TC.
MS on LC performed worse, and a significant disparity between cost functions was observed, with TC and PC achieving reductions of $33\%$ and $28\%$, respectively, likely due to cancellations not accounted for with the JW mapping.

For the Fermionic pool on FC, we observe an improvement of $25\%$ and $19\%$ for TC and PC, respectively.
With MS, TC outperformed PC with a $27\%$ and $25\%$ reduction, respectively. 
On LC, we saw improvements of $28\%$ and $25\%$ for TC and PC, with no discernible difference between the CP and NCP search spaces. 
For MS, we see a $33\%$ and $21\%$ reduction in TC and PC, respectively.
Now, the PC cost does not faithfully represent the resulting CNOT cost.
With this pool, we find MS generally performs best.
We partially attribute this success to the ability to use efficient circuit compilation passes with these mappings. 
The PC cost did not faithfully represent the resulting CNOT cost for the Fermionic pool on LC, likely due to the compilation scheme of the Fermionic Pauli generators not sufficiently accounting for LC constraints and CNOT cancellations.

In Figure~\ref{fig:molecular results} and Table~\ref{tab:molecular results}, we present the best-performing setting and cost function of treespilation, showcasing the resulting CNOT counts. 
We achieved an overall average improvement of $28\%$ and $49\%$ for the Fermionic and Majoranic pools compared to the initial JW encoded ansatz on FC.
The largest reduction was for the Majoranic pool with LiH, reaching three-quarters.
For the Sycamore device, we report a respective average improvement of $34\%$ and $52\%$, and for the Eagle, $37\%$ and $52\%$. 

For reference, we include counts of the QEB and qubit pool ansatz. 
On FC, the best-treespiled result exhibits an average improvement of $23\%$ compared to the best performance of the QEB and qubit pools.  
For LC, the improvement is $44\%$ and $51\%$ for Sycamore and Eagle connectivities, respectively.
We observe that CNOT counts on the Sycamore device are lower than those of the Eagle. 
This is pronounced in the case of the Fermionic and QEB pools, where the higher average degree of connectivity of each qubit makes transpiling these operators easier, as they demand a high degree of connectivity to implement efficiently.

\begin{figure*}[t]    
\centering
    \includegraphics[width=.9 \linewidth]{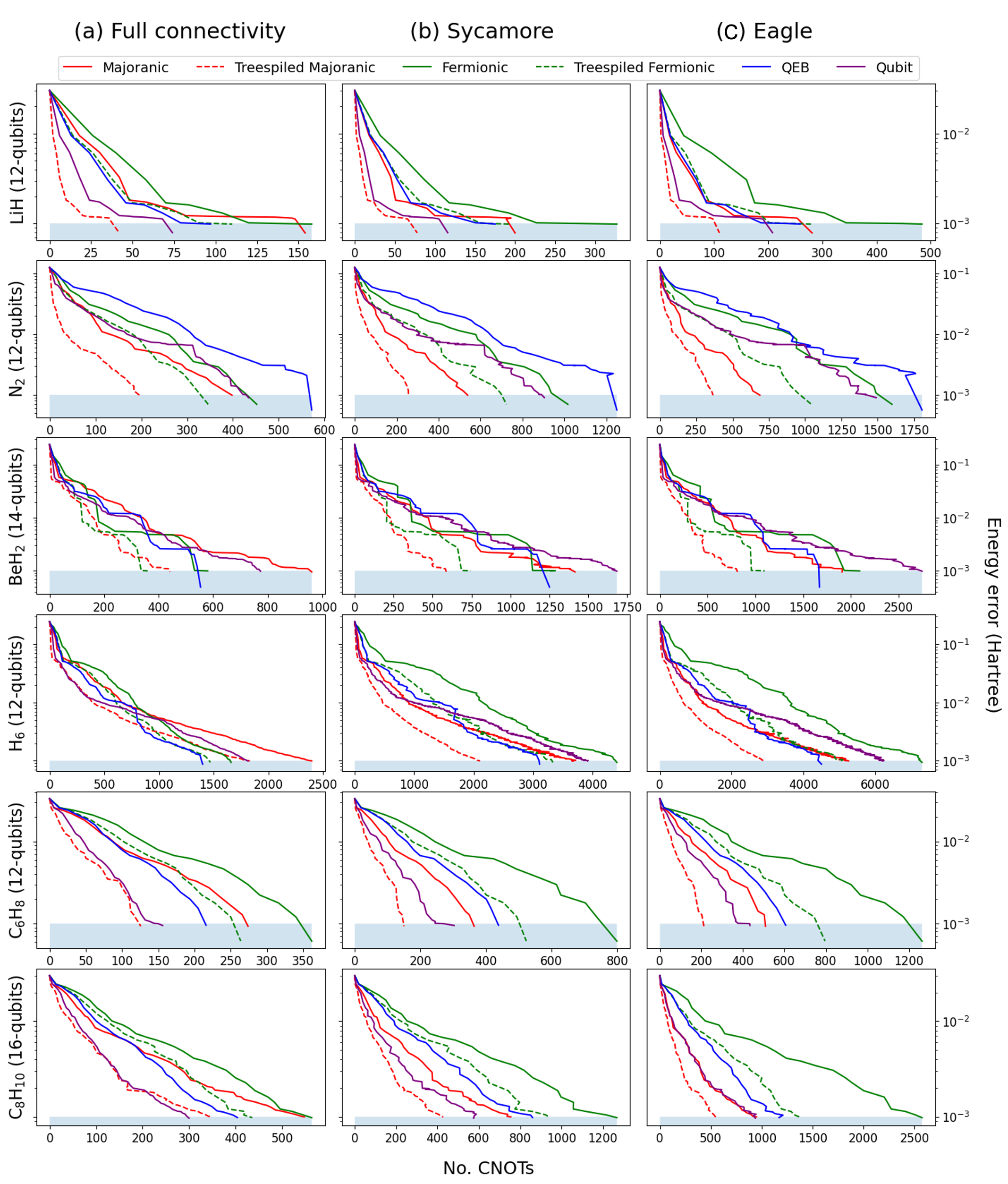}
    \caption{
    CNOT convergence of ADAPT-VQE groundstate simulations across six molecules, considering full, sycamore, and eagle connectivity.
    For each molecule, six curves are presented: one for QEB- and qubit-pools, and one each before and after treespilation with the Fermionic- and Majoranic pools.
    Before treespilation, the ansatz is encoded using the Jordan-Wigner transformation. 
    The dashed lines in the treespilation curves represent points where the optimized mapping from the previous point serves as the starting point for the algorithm, and the resulting CNOT count is plotted.
    Due to the stochastic nature of the qiskit transpiler used, the lines do not exhibit monotonic convergence.
    }
    \label{fig:ADAPT convergence cnots}
\end{figure*}

We found that treespilation shifted the aforementioned tradeoff in favour of the Fermionic pools, resulting in improved performance compared to their non-Fermionic counterparts. 
One might attribute this to the notion that the Fermionic and Majoranic pools generally have fewer elements in their ansatz. 
However, this is not entirely the case as, for example, with LiH where the number of ansatz elements was identical, the Majoranic pool outperforms the qubit pool.

We analyzed the performances across the entire energy convergence against CNOTs of the ansatz in Figure~\ref{fig:ADAPT convergence cnots}. 
We opted for unrestricted, and CP search spaces with PC for the Majoranic pool on both FC and LC, as we noted satisfactory performance with these selections.
Selecting CP ensures that the resulting mapping for quadratic- and quartic-Majorana-products, represented by the pool elements in qubit space, can be optimally implemented in polynomial time, as demonstrated in Appendix~\ref{app:steiner-tree}.
MS with TC performed best for the Fermionic pool, so we highlight this setting and cost. 
The treespiled Majoranic pool performed exceptionally well on LC, often outperforming the initial JW-encoded FC result. 
On the other hand, the qubit pool's performance was underwhelming partially due to more generators in the ansatz.
Another factor is the increased necessity for costly SWAP gates to interconnect the sparse generators on LC.
This contrasts the treespiled Majoranic pool, where we use the mappings tree structure to reduce this sparsity, replacing it with Pauli operators on qubits where SWAPs would have been necessary.
For example, with the CP setting the mapping tree is connected on the device with nodes representing physical qubits, causing the resulting Pauli string generators to act on `mostly connected' qubits.
Moreover, we cannot perform the efficient Steiner-tree pass, as we do not know the optimal layout of this pool \textit{ab-initio}. 

Overall, we find that using proper Fermionic pools can significantly reduce the number of CNOTs required for implementing ansatz compared to non-Fermionic pools. 
Treespilation provides an effective approach to achieve these reductions, particularly for the Majoranic pool on LC. 
However, more efficient compilation methods are needed for the Fermionic pool on LC.

\section{Conclusion}
In this paper, we have presented a Fermion-to-qubit mapping scheme to reduce the number of CNOT gates required to implement Fermionic states through quantum circuits. 
To achieve this, we defined a space of ``good" mappings characterized by product-preserving ternary tree-based mappings combined with Majorana braidings within Fermionic mode pairs. 
Additionally, we introduced fundamental tree-mapping transformations allowing for the deformation of any mapping within this class, along with a procedure for optimization within this space.
Furthermore, we establish cost functions that quantify a mapping's CNOT cost within this space.
Using these tools, we introduce the ``treespilation" algorithm that tailors a mapping to the Fermionic representation of the ansatz while considering the potentially limited-connectivity architecture of the quantum device. 
Essentially, our method can be seen as a meta-compilation approach that augments the results from a given transpilation scheme by optimizing the qubit representation of a Fermionic state.

To illustrate our approach, we applied it to ansatz generated through statevector ADAPT-VQE simulations, representing ground-state approximations of molecules with different Fermionic-based operator pools. 
In summary, our method yielded encodings of Fermionic states that significantly reduced the number of CNOTs required to represent the qubit state on both full and limited connectivity quantum computers.
For instance, in the case of LiH on full connectivity, we observed a remarkable 74\% reduction in CNOTs compared to the initial Jordan-Wigner encoded ansatz.
Additionally, when comparing our method to some similar, CNOT-efficient, non-Fermionic-based pools, we found that, on average, our approach significantly outperforms them.

In summary, our scheme presents a promising approach to reducing the CNOT requirements for implementing Fermionic ansatz on full and limited connectivity quantum computers. 
Using it, we show that we can essentially eliminate much of the circuit burden usually incurred by the mapping of many-body Fermionic states to qubits.
Furthermore, we anticipate that the tools introduced in this paper may be utilized to optimize various aspects of Fermionic simulations, including measurement cost, state fidelity, circuit depth and so on. 
It's important to note that while we use ADAPT-VQE as an example, our methodology is broadly applicable in optimizing the representation of Fermionic unitaries, such as chemical Hamiltonians. 
Further work may explore more advanced optimization schemes than simulated annealing.
Here, we explore mappings where the number of modes is equal to the number of qubits simulated, however, it would be interesting to study the effect of introducing ancillary and removing qubits through symmetries. 

\section*{Acknowledgements}
The authors thank Anton Nykänen for producing the ansatz studied in this paper, and Guillermo García-Pérez for fruitful discussions.

\section*{Competing interests}
Elements of this work are included in patents filed by Algorithmiq Ltd with the European Patent Office.

\section*{Author contributions}
AM conceived the method. AM, AG, and ZZ  initiated and planned the research. AM implemented the algorithm and made numerical analysis. AM and AG were responsible for the analytical investigations presented in the Appendices. AM wrote the first version of the manuscript. AM, AG, and ZZ  contributed to scientific discussions and to the writing of the manuscript.

\section*{Additional Information}
The \textit{treespilation} algorithm is implemented as part of \textit{Aurora}’s suite of algorithms for chemistry simulation. 
Work on ``Quantum Computing for Photon-Drug Interactions in Cancer Prevention and Treatment" is supported by Wellcome Leap as part of the Q4Bio Program.

\onecolumngrid

\appendix

\section{Polynomial-time algorithm for Steiner-based implementation of Pauli exponentiation \label{app:steiner-tree}}
Using the results from~\cite{Vandaele2022}, one can find an efficient implementation of arbitrary exponentiation $\exp(-\mathrm i tP)$ for any real $t$ and any Pauli string $P$ on limited connectivity defined through a connected graph $G=(V,E)$. 
First, we apply one-qubit gates to transform the operation $\exp(-\mathrm i tP)$ into $\exp(-\mathrm i tP_Z)$, where $P_Z$ is a Pauli string defined over $\{I, Z\}$.  
The transformed Pauli, $P_Z$, acts non-trivially on the same set of qubits $V_P \subseteq V$ as $P$.
We proceed then exactly as described in~\cite{Vandaele2022}. 
A Steiner tree $T = (V_T, E_T)$ of $V_P$ over $G$ is found, which is the minimal in the number of edges subgraph of $G$ such that $V_T \supseteq V_P$. 
This generates the minimal set of qubits that need to interact to implement $\exp(-i tP_Z)$. 
Then, the following operations are applied iteratively:
\begin{enumerate}
    \item we choose a particular leaf $v$ of the tree and a vertex $w$ connected to it,
    \item we apply CNOT with control on $v$ and target on $w$; if $w\not\in V_P$, in addition, we apply CNOT with control on $w$ and target on $v$,
    \item we remove $v$ from $V_T$, $\{v,w\}$ from $E_T$, and we repeat the above steps.
\end{enumerate}
The procedure ends when we are left with just one node $v'$, on which we apply the rotation $\exp(-i t Z_{v'})$. Finally, we uncompute all the steps done before the single-qubit rotation.
The procedure requires $2(2 V_T - V_P - 1)$ CNOTs.

Thus, the difficulty of finding a CNOT efficient circuit is reduced to the problem of finding the Steiner tree, which is known to be NP-hard~\cite{Karp1972}. However, if the ternary tree of PPTT-based F2Q is a subgraph of the connectivity graph $G$, one can show that the problem of finding Steiner trees for any product of a Pauli string resulting from the product of 2 or 4 Majoranas can be done efficiently. 
These cases are important as the resulting strings compose the Fermionic single- and double-excitation operations described in the main text and also form the Pauli strings in the second quantised Hamiltonian.
The rest of the section is dedicated to formally proving this fact. 
We start by introducing a variant of the Steiner tree problem. 
We adopt the notation $G[V']$ for the induced subgraph of $G$ generated by the vertex subset $V'$.

\newcommand{\ppttsteinertree}{\textsc{PPTT F2Q Steiner Tree Problem}\xspace}
\begin{definition}[\ppttsteinertree] Let $G=(V,E)$ be an undirected connected graph, and $V' \subseteq V$ be a set of terminals. We call finding a Steiner tree of $V'$ in $G$ a \ppttsteinertree if at least one of the following holds:
\begin{enumerate}
    \item $G[V']$ has at most two connected components, or
    \item $G[V']$ is disconnected and there is exactly one vertex $v \in V \setminus V'$ s.t. the $G[V' \cup \{v\}]$ is a connected component
\end{enumerate}
\end{definition}

Let us start by showing that it is easy to solve the \ppttsteinertree exactly.
\begin{theorem}\label{theorem:pptt-in-p}
\ppttsteinertree is in \textbf{\emph P}.
\end{theorem}
\begin{proof}

Suppose that $G[V']$ is connected. Then it is enough to find any spanning tree of the graph, which can be done in polynomial time. 

Suppose that $G[V']$ consists of two connected components, defined over vertex sets $V'_1$ and $V'_2$. Let $(v,v_1,\dots,v_k,w)$ be the shortest path over $(v,w)\in V'_1 \times V'_2$. Such a path can be found by applying e.g. Floyd-Warshall algorithm that runs in $\order{|V|^3}$ time, and exhaustively searches for the smallest distance which takes $\order{|V'_1|\cdot | V'_2|} = \order{|V|^2}$. Note that $v_i \in V\setminus V'$, as otherwise shortest path could be found that connects vertices from $V'_1$ and $V'_2$. Note that adding edges from this path to $G[V']$ makes a new connected graph, for which we can find a spanning tree efficiently. By the definition, such a graph is the Steiner tree of $V'$. On the other hand, one cannot hope for finding a smaller Steiner tree as it would contradict the chosen path $(v,v_1,\dots,v_k,w)$ to be the shortest path.

Finally, if $G[V']$ is disconnected up to one vertex $v$, it is clear that the spanning tree of $G[V' \cup \{v\}]$ is a minimal Steiner tree for $V'$. 
\end{proof}

In addition, we can show that adding edges to the graph $G$ does not make the problem harder.
\begin{theorem} \label{theorem:monotonicity-ppttsteinertree}
    Suppose we're given a \ppttsteinertree defined over graph $G=(V,E)$ and set of terminals $V' \subseteq V$. Then adding any new edge to $E$ makes it still an instance of \ppttsteinertree.
\end{theorem}
\begin{proof}
    Let $G'$ be the new graph formed by adding a new edge to $G$. It is enough to show that $G'[V']$ satisfies the same properties as $G[V']$ required by the definition of \ppttsteinertree. 
    
    If $G[V]$ has at most two connected components, then the same can be said for $G'[V']$. Furthermore, if there is a unique vertex $v$ s.t. $G[V' \cup \{v\}] $ is connected, adding an edge will either make $G'[V']$ a connected subgraph which makes the problem of the first type as introduced in the definition of \ppttsteinertree, or $G'[V' \cup \{v\}]$ is connected. 
\end{proof}
The practical implication of the theorem above is that demonstrating the polynomial complexity of finding a Steiner tree on the ternary tree implies its polynomial complexity on the hardware connectivity graph. 
It is essential to note that this doesn't guarantee the Steiner tree will be the same, as increasing the number of edges in the hardware connectivity graph may allow finding a smaller Steiner tree.

We conclude this section by demonstrating that, for any Pauli string representing the product of 2 or 4 Majoranas, solving the Steiner tree problem can be accomplished in polynomial time. 
The proof will be presented by establishing that the Steiner problem will conform to the form defined in \ppttsteinertree. We begin with the product of 2 Majoranas, which can be shown to consistently form a connected path on the ternary tree in qubit space.

\begin{theorem} \label{theorem:majorana-product-2}
    Let the hardware connectivity graph $G$ be a connected graph. Let there be an arbitrary PPTT F2Q defined through a ternary tree which is a subgraph of $G$. Then for any product of two strings generated by the PPTT mapping, the qubits on which the products act non-trivially induce a connected component on $G$.
\end{theorem}
\begin{proof}
    Let $S_1$ and $S_2$ be two Majorana strings (i.e., Pauli strings representing Majorana operators) generated by the PPTT mapping as described in the theorem statement. Each is constructed by taking a path from the root to a leg $(v_1 = \text{root}, \dots, v_k, l)$, where $v_k$ is the last node to which the leg $l$ is assigned. This path defines each Majorana string. For each node $v_i$ with the label $(j, u, c)$, we take the $v_{i+1}$-th Pauli label $P\in \{X,Y,Z\}$ and apply it to qubit $j$. The product over $v_1, \dots, v_k$ (for the last one, we take the leg's Pauli operator) defines the Pauli string.
    
    Now, suppose that $v_2$ is different for both Majorana strings, indicating that the strings \emph{diverge} at the root $v_1$. This implies that there is no shared ternary tree node for the strings except for the root. In addition, the Pauli operator acting on the root is different. Since for any two non-identity Paulis $P_1$ and $P_2$, we have $P_1P_2 \neq I$, the product of these two Majorana strings will be a path of non-identity Paulis acting on qubits from the leg of one string, through the root, to the leg of another string. Such a Pauli operator acts on qubits forming a connected path on the mapping's underlying tree, and hence a connected path on $G$.
    
    Note that if $v_2$ is the same for both strings, then there must be a node $v_m$ (possibly a leaf) in the path at which the strings diverge. In this case, the product of those strings does not include the root, as we act with the same Pauli on it, and on all the nodes up to $v_m$ exclusively. However, in this case, the strings form a connected path from one leg to another through $v_m$ to the leg of another string. In this scenario, we can also observe that the qubits form a connected (possibly one-node) path.
    
\end{proof}
Note that in the proof, we utilized the fact that if two Majorana strings diverge at any point, they form a connected component. 
This is because the two strings are identical up to the point of divergence and cancel out when multiplied together. 
This fact will be frequently used in the following theorem for the product of 4 Majorana strings. 
We adopt the notation $A \propto B$ if $A = cB$ for some complex number $c$.

\begin{theorem}\label{theorem:majorana-product-4}
    Let the hardware connectivity graph $G=(V,E)$ be a connected graph. 
    Assume there is an arbitrary PPTT mapping defined on a ternary tree that is a subgraph of $G$. 
    Let $S$ be the product of any 4 Majorana strings. 
    Then, finding the Steiner tree over $G$ for qubits on which $S$ acts non-trivially can be done in polynomial time.
\end{theorem}
\begin{proof}
    The proof will proceed by showing that finding a Steiner tree on the ternary tree is a \ppttsteinertree.

    We will prove the statement by considering several scenarios in which Majorana strings will diverge at different nodes and in various combinations of directions.
    \begin{enumerate}
        \item \textbf{Majorana strings leave the root in all $X$, $Y$, $Z$ directions:} In this case, the root is a terminal as $PPP'P''\propto P$ for any pairwise different non-identity $P,P',P''$. With this example, all the nodes that are on the paths of Majorana strings leaving into $P'$ and $P''$ directions form a connected path that passes the root and are terminals. On the other hand, for Majorana strings leaving into $P$ direction, the case reduces to what was observed for the product of two Majorana strings analyzed in Theorem~\ref{theorem:majorana-product-2}. This means that all the nodes up to a diverging node (excluding the diverging node and the root) will not be terminals, and all the other nodes from these two Majorana strings will be terminals. Eventually, we will obtain a single connected component if the Majorana strings diverge at the $P$-children of the root, or they will form independent connected components if they diverge later, reducing the problem to \ppttsteinertree.
        \item \textbf{Majorana strings leave the root in pairs in two directions:} In this case, the root is not included as $PPP'P' = I$ for any Paulis $P,P'$. Two pairs of Majorana strings will always form two disconnected components, following the reasoning in Theorem~\ref{theorem:majorana-product-2} for each pair independently, again reducing the problem to \ppttsteinertree.
        \item \textbf{Three Majorana strings leave the root in one direction, and the 4th one leaves in another direction:} In this case, the root is always a terminal, as $PPPP' \not\propto I$ for any $P,P'$ different than identity. Similarly, all the nodes on the path before the next diverging point for the three Majorana strings will be terminals until they diverge. At the diverging qubit, one of the two possibilities can occur:
        \begin{enumerate}
            \item \textbf{All Majorana strings leave in different directions:} The diverging node is not a terminal as $XYZ \propto I$, however, all other nodes are terminals, forming in total 4 connected components including the one already created with the root. In this case, adding the diverging node will make the induced graph connected, which reduces the problem to \ppttsteinertree.
            \item \textbf{Two Majorana strings leave in one direction, and the 3rd Majorana string leaves in another direction:} In this case, the diverging node is a terminal as $PPP' = P'$ for any non-identity $P, P'$, and so are the nodes from the 3rd Majorana string. The two Majorana strings going in the same direction might form the second connected component if they do not diverge at the children of the diverging node. Since eventually we have at most two connected components, the problem reduces to \ppttsteinertree.
        \end{enumerate}
        \item \textbf{All Majorana strings leave in the same direction:} In this case, neither the root nor the nodes on the path before the first diverging node are terminals as $P^4 = I$ for any Pauli $P$. Then, the analysis reduces to one of the cases above except instead of the root, the diverging node is a starting point.
    \end{enumerate}
    As shown in the analysis above, in all cases, the problem reduces to \ppttsteinertree. By Theorem~\ref{theorem:monotonicity-ppttsteinertree}, the result generalizes to a hardware connectivity graph $G$ after adding missing nodes from $G$ that are not in the ternary tree, allowing finding the Steiner tree in polynomial time thanks to Theorem~\ref{theorem:pptt-in-p}.
\end{proof}

The two last theorems above can be concluded with the following lemma.
\begin{theorem}
    Let $S$ be a Pauli string, representing a product of 2 or 4 Majorana strings defined over $k$ qubits, coming from the PPTT F2Q mapping with the ternary tree as a subgraph of the hardware connectivity graph. One can find an optimal Steiner tree with $n \geq k$ nodes in the hardware connectivity graph in polynomial time, which, in turn, allows the implementation of $\exp(-itS)$ for any real $t$ with only $2(2n-k-1)$ CNOTs.
\end{theorem}
\begin{proof}
    The fact that the Steiner tree can be found in polynomial time comes directly from Theorems~\ref{theorem:majorana-product-2} and \ref{theorem:majorana-product-4}. The number of CNOTs comes from the implementation proposed in~\cite{Vandaele2022}.
\end{proof}

\section{Reachability of PPTT F2Q mappings} 
\label{sec:reachability}
In this section, we analyze whether the PPTT F2Q transformations, as presented in Sec.~\ref{sec:methodology}, allow transforming any hardware connectivity-preserving PPTT F2Q to any other hardware connectivity-preserving PPTT F2Q. Starting from now, we will assume that we are given an undirected graph $G=(V,E)$ representing the quantum hardware connectivity, and the ordered ternary tree (OTT) used for any considered PPTT F2Q is a subgraph of $G$ without explicitly stating it.

For general graphs, it can be shown that transforming one F2Q mapping into another is not always possible with the transformations from Sec.~\ref{sec:methodology}. Consider a full OTT. For sufficiently many nodes (at least 16), there is a node $v$ that has both a parent and all three children, none of which are leaves. Now, suppose that the graph $G$ is the underlying ternary tree, and to $v$ we attach a long path tree only. The aforementioned F2Q mapping, as well as, for example, the JW mapping defined on this path graph, are hardware-connectivity preserving PPTT F2Q mappings. However, transforming the former to the latter would require creating a tree with one node connected to 5 nodes at once, as shown in Fig.~\ref{fig:reachability-counterexample}. Since our PPTT F2Q mappings require the tree to be a ternary ordered tree, such a move is not allowed. Note that if we relax this restriction that the OTT must be a subgraph of $G$, then of course we can reach any PPTT mapping.

\begin{figure}[t]
    \centering
    \includegraphics[width=0.42\linewidth]{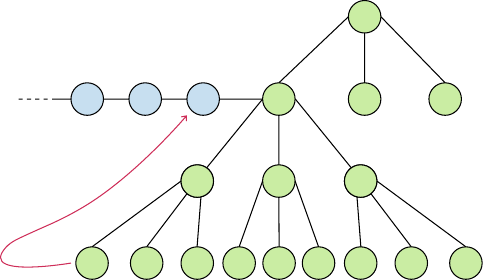}
    \caption{An example of the ternary tree (green nodes) which cannot produce any other ternary tree. The only node to which any leaf could be attached already has degree 4 in the tree, and degree 5 is not allowed}
    \label{fig:reachability-counterexample}
\end{figure}

On the other hand, it is possible to provide sufficient conditions that allow showing that any PPTT mapping can be generated from any other on a sufficiently large heavy-hexagonal and two-dimensional grid graph. We dedicate the rest of the section to proving this fact. We start with the following lemma, which reduces the complexity of changing between F2Q mappings, as long as the rules allow changing any subtrees with at most degree 3. For this, we introduce the term ``underlying tree." Let $H$ be the OTT used for a particular PPTT F2Q. The underlying tree of $H$ is a simple graph (without a root pointed) such that all the arcs in the OTT are replaced with edges. Note that this is a proper subgraph of $G$ and a tree. The underlying tree of an F2Q mapping is the underlying tree of the OTT used in the mapping.

\begin{lemma} \label{lemma:ignore-all-but-tree}
Let $G=(V,E)$ be an undirected connected graph, and let there be two PPTT F2Q mappings $F_1,F_2$ with underlying trees $T_1, T_2$. Assuming that moving leaves allows transforming $T_1$ into $T_2$, one can transform $F_1$ into $F_2$ with the steps introduced in Sec.~\ref{sec:methodology}.
\end{lemma}

\begin{proof}
First, let us note that for any fixed OTT, swapping modes associated with any two nodes allows the production of any mode association. Similarly, changing braiding can also be done independently of the other changes to the PPTT F2Q mappings. Moreover, children for any node can be reassigned with a new Pauli without changing the underlying tree. Finally, since one can change any root to any other node that can be a root, the only issue with changing one F2Q to another is with changing the underlying tree.
\end{proof}

Note that the lemma almost allows us to focus solely on the underlying trees. However, at this moment, it remains unclear if we can move leaves freely in the underlying tree, as we may accidentally connect the 4-th node to the root of the tree. Fortunately, by changing the root, we don't have to worry, as one can attach a node `to the root'.

\begin{lemma} \label{lemma:move-one-node}
    Let $T_1$ be an underlying tree of OTT, and let $T_2$ be another underlying tree of OTT that differs from $T_1$ by reattaching the leaf. Then it is possible with the steps introduced in Sec.~\ref{sec:methodology} to reach PPTT F2Q with underlying tree $T_2$ out of PPTT F2Q with underlying tree $T_1$.  
\end{lemma}
\begin{proof}
Note that the only problem that arises might be if we would assign the moved leaf to a root and make it a degree-4 node. However, such a node can't be a root for $T_2$, so before moving the leaf we simply have to move the root to any node which is not the moved leaf. Fortunately, one can show that for any tree, one can find at least two nodes with degree 1, as otherwise we contradict the degree sum formula:
\begin{equation}
    \sum_v \text{degree}(v) \geq 2(n-2) + 1 > 2(n-1) = 2|E|.
\end{equation}
In light of Lemma~\ref{lemma:ignore-all-but-tree} which allows us to change mode association, braiding and root position arbitrarily, we have proved the statement of the lemma.
\end{proof}

The lemmas above allow us to think about the reachability of any PPTT F2Q mapping from any other PPTT F2Q only in terms of underlying graphs, without even being concerned about the root position. 
Before showing sufficient conditions for reaching any PPTT F2Q mapping, let's introduce a definition that will turn out to be useful for demonstrating how to construct an arbitrary underlying tree on particular hardware connectivity graphs.

\begin{definition}
    Let $G=(V,E)$ be arbitrary undirected connected graphs and let $V' \subset V$. Let $W\subset V$ s.t. $G[W]$ and $G[V\setminus W]$ are connected, and $W \cap V' \neq \emptyset$. We call sequence of sets $(V_1,\dots,V_k)$ a ping-pong partition of the $V'$ over $(G,W)$ if simultaneously
    \begin{enumerate}
        \item $\{V_i\}_i$ is partition of $V'$,
        \item $V_i \subseteq W$ for odd $i$ and $V_i \subseteq V\setminus W$ for even $i$,
        \item $G[\bigcup_{i=1}^{k'} V_i]$ is connected for any $1 \leq k' \leq k$.
    \end{enumerate}
\end{definition}
Note that \textit{3.} implies that $G[V']$ is connected. A particular algorithm for generating one goes as follows: first, we choose an arbitrary vertex $v\in V' \cap W$ and find the maximum number of nodes in the induced graph inside $G[W]$. The vertex set of this induced graph is our $V_1$. Then we find the maximum set of nodes $V_2\subseteq V\setminus W$ such that $G[V_1 \cup V_2]$ is connected. We repeat the process, constructing consecutive $V_1,\dots, V_k$ until they form a partition of $V'$.

Equipped with such definitions, we are ready to prove the main theorem of this section.
\begin{theorem} \label{theorem:reachability}
    Let $G=(V,E)$ be a simple connected graph and suppose we are given a $n$-mode PPTT F2Q mapping defined over $G$. Then one can create any other PPTT F2Q mapping provided the sufficient conditions on $G$ hold:
\begin{enumerate}
    \item the maximum degree of $G$ is 4,
    \item there exists $W\subset V$ s.t. $G[W]$ and $G[V\setminus W]$ are connected and $|W|, |V\setminus W| \geq n$.
\end{enumerate}
\end{theorem}
\begin{proof}
    Let there be two PPTT mappings with underlying trees \(T_1, T_2=(V_{T_2},E_{T_2})\). In light of Lemma~\ref{lemma:ignore-all-but-tree} and Lemma~\ref{lemma:move-one-node}, it is enough to show that by moving leaves such that all intermediate graphs are subtrees, we can transform tree \(T_1\) into \(T_2\). Note that using the fact that the maximum degree of nodes in \(G\) is 4, we can never produce a tree with a maximum degree of 5 or more. Thus, we can always find a valid PPTT mapping for each of them.
    
    Without loss of generality, we will assume that \(T_1\) has nodes in \(V\setminus W\) only, and \(T_2\) has nodes in both \(V\setminus W\) and \(W\). This is because if we can transform \(T_1\) into \(T_2\), we can also transform \(T_2\) to \(T_1\). Additionally, the same conclusions by symmetry can be done by swapping \(V\setminus W\) and \(W\). Finally, such transformations can be chained to eventually allow the transformation of any tree to any other tree.
    
    Furthermore, we will assume that \(T_1\) has a node that has a neighbour in \(W\cap V_{T_2}\), a so-called element of the border of \(G[V\setminus W]\). Otherwise, we could transform the tree so that it will satisfy this assumption as follows. First, we look for a path over nodes from \(V\setminus W\) that connects a particular node in \(T_1\) with any node \(v\) from the border such that intermediate nodes are not in \(T_1\). Finally, we iteratively move leaves from the consecutively generated trees and add them along the path, up to \(v\) inclusive. This way we can transform \(T_1\) appropriately.
    
    Now let \((V_1,\ldots V_k)\) be a ping-pong partition of \(V_{T_2}\) over \((G,W)\). Given \(T_1\), we create a tree \(T^{(1)}\) by moving leaves from \(T_1\) to \(V_1\). Note that since the border element \(v\) is a neighbouring element of \(W \cap T_{V_2}\), each leaf moved from \(T_1\) can be already assigned to some element from \(V_1\), connecting them according to \(T_2\). This way, after moving \(|V_1|\) nodes, we have a subtree with all the nodes from \(V_1\) and edges as in \(T_2[V_1]\). Finally, we move all the remaining nodes from \(T_1\) and assign them arbitrarily to \(V_1\) so that they will form a connected tree \(T^{(1)}\) with all nodes in \(W\) and \(T^{(1)}[V_1] = T_2[V_1]\). Note that since \(G[W]\) is a connected subgraph with \(n\) nodes, one can always find free nodes in \(W\setminus V_1\).
    
    The steps above are repeated for \(V_{k'}\) with \(k'=2,\dots,k\) with the following updates:
    \begin{enumerate}
        \item We are not moving nodes that are already in \(\bigcup_{i=1}^{k'-1} V_i\).
        \item The nodes from \(V_{k'}\) should be connected to \(\bigcup_{i=1}^{k'-1} V_i\) along the \(T_2\) structure.
    \end{enumerate}
    Constructing a new \(T^{(k')}\). With such rules, for each \(k'\), we have \(T^{(k')}[\bigcup_{i=1}^{k'} V_i] = T_2[\bigcup_{i=1}^{k'} V_i]\), which eventually will produce \(T_2\) for \(k'=k\). The visualization of the process described in the last two paragraphs can be found in Fig.~\ref{fig:reachability-example}.
     
\end{proof}

\begin{figure}
    \centering
    \includegraphics[width=0.8\linewidth]{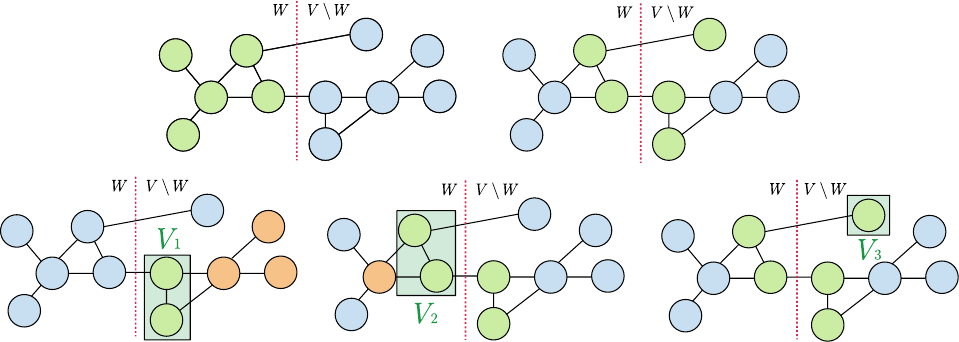}
    \caption{A visualization of the proof presented in in Theorem~\ref{theorem:reachability}. We start with the F2Q mapping which is fully on $W$ (top left), and we wish to obtain the F2Q mapping as on top right. This can be done in three steps as depicted below, where the green nodes are nodes that are already in the correction location in the tree, and the orange ones are yet-to-be-moved nodes.}
    \label{fig:reachability-example}
\end{figure}

Note it is rather easy to relax the conditions of the theorem. First, we don't need a partition into $W$ and $W\setminus V$ both having $n$ or more nodes as long as we can fit all the remaining nodes in the process. With similar arguments, we don't necessarily require $G[W]$ and $G[V\setminus W]$ to be connected as already depicted in Fig.~\ref{fig:reachability-example}. Finally, the condition on maximum degree could also be relaxed as long as we can guarantee that the intermediate trees will always have maximum degree 4.

However, this theorem is sufficient to show that reaching any PPTT F2Q mapping is possible on sufficiently large heavy-hexagonal and 2D grid graphs. For both classes, the maximum degree of nodes are 3 and 4 respectively, and since they are two-dimensional structures, it is easy to split them vertically or horizontally into two connected halves, see Fig.~\ref{fig:connectivity graphs} for the arrangement of qubits.

\section{Estimates on the number of PPTT F2Q mappings}
\label{sec:number-of-f2q}

In this section, we estimate the number of possible PPTT F2Q mappings for a particular hardware connectivity graph \(G\). Two scenarios will be considered: in the first, we will assume \(G\) is a complete graph; in the second, we will assume the maximum degree is bounded by 3.

\paragraph{A complete graph \(G\)} Let \(n\) be the number of modes, and \(Q\) be the number of qubits. First, note that mode associations and braiding can be chosen arbitrarily, giving us \(2^n n!\) possibilities. Now, let's count the number of ternary trees available. Starting with a root, we can attach a new node to any of its 3 possible children. Then we can attach a new leaf in 5 possible ways, and so on. In total, the number of ordered ternary trees is at most \(\prod_{i=1}^n (2i-1) = (2n-1)!! = (2n)!/(2^{n} n!)\). Additionally, we can assign any of the separately chosen \(\binom{Q}{n}\) physical qubits and assign them to nodes in any of \(n!\) ways. Note that with the procedure above, some ternary trees can be constructed in more than one way. This gives us an upper bound on the number of PPTT F2Q mappings as

\begin{equation}
    2^n n! \binom{Q}{n} n! \frac{(2n)!}{2^{n} n!} = n! (2n)! \binom{Q}{n} \leq n! (2n+1)! \frac{Q^n}{n!} = Q^n (2n+1)! = 2^{n (\log Q + \log n - \log(e)) + \order{\log (n)}},
\end{equation}

where, at the end, we used Stirling's formula \(\log(n!) = n\log n - n \log e + \order{\log n}\). Note that for indistinguishable qubits, we can just choose \(Q=n\), which simplifies the formula to

\begin{equation}
    2^{n (2\log n - \log(e)) + \order{\log _2(n)}  }.
\end{equation}

\paragraph{A bounded-degree graph \(G\)} Here, we assume the maximum degree of the graph is \(d>2\). All the steps are as before, except for how many trees we can find. In the \(i\)-th step, instead of having \(2i-1\) possibilities of assigning a node and an arbitrary physical qubit to be attached, we have to consider only those qubits which are neighbouring.

Therefore, first, we choose a root among one of \(Q\) qubits. Then, for one of the 3 legs, one of \(d\) neighbours of the root is chosen. Since all the used nodes in the trees can have at most \(d-1\) neighbours, we can upper bound the number of possibilities a physical qubit can be attached as the \(i\)-th node for \(i = 3,\dots, n\) as \(3 (d-1) (i-1)\), where 3 comes from upper bounding the number of free legs for each node, \(d-1\) from the number of nodes that can be attached to each node in the tree, and \(i-1\) from the total number of nodes in the current tree. Therefore, we can construct at most
\begin{equation}\label{eq:num-trees-bounded-graphs}
    Q \cdot 3d \cdot \prod_{i=2}^{n} 3(d-1)(i-1)  = Q 3^n d (d-1)^{n-1} \prod_{i=1}^{n-1} i \leq Q 3^{n} d^n (n-1)!.
\end{equation}

Thus, the total number of PPTT F2Q mappings is
\begin{equation}
    Q 3^{n} d^n (n-1)! 2^n n! \leq Q (6d)^{n} (n!)^2 = 2^{n (2\log n + \log (6d/e^2))  + \order{\log (Qdn)}},
\end{equation}
where we again used Stirling's formula. Note that the dependency on the number of qubits \(Q\) is essentially lost, suggesting that for bounded-degree graphs like the heavy hexagonal with \(d=3\) or 2D grid with \(d=4\), the number of F2Q mappings is not significantly larger (if larger at all) than for the complete graphs with indistinguishable qubits. Note that here, all the qubits are assumed to be indistinguishable, so apart from overestimating the number of trees in Eq.~\eqref{eq:num-trees-bounded-graphs}, we did not account for possible isomorphism between ternary trees. Furthermore, for heavy-hexagonal trees, many of the nodes have a degree of 2, which further decreases the number of different PPTT F2Q mappings.

\section{Simulation data and Results \label{app:molecular data}}

\begin{figure*}[h!]
    \centering
    \includegraphics[width=0.8\textwidth]{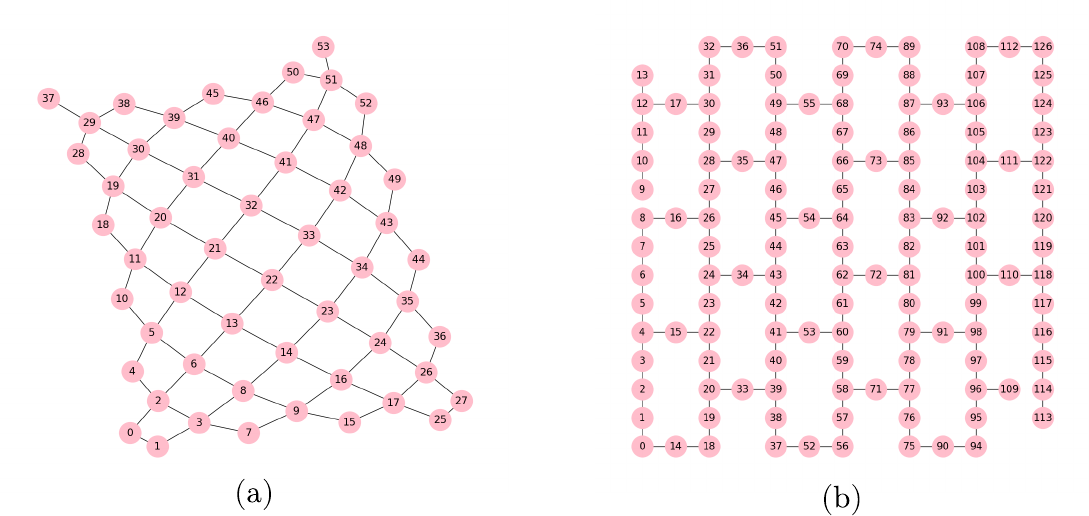}
    \caption{Figure (a) illustrates the topology of the IBM Washington device, and $(b)$ is the Google Sycamore. CNOT gates between qubits are allowed only along the edges of the graphs in these devices.}
    \label{fig:connectivity graphs}
\end{figure*}

\begin{figure*}[h!]    
\centering
    \includegraphics[width=0.95 \linewidth]{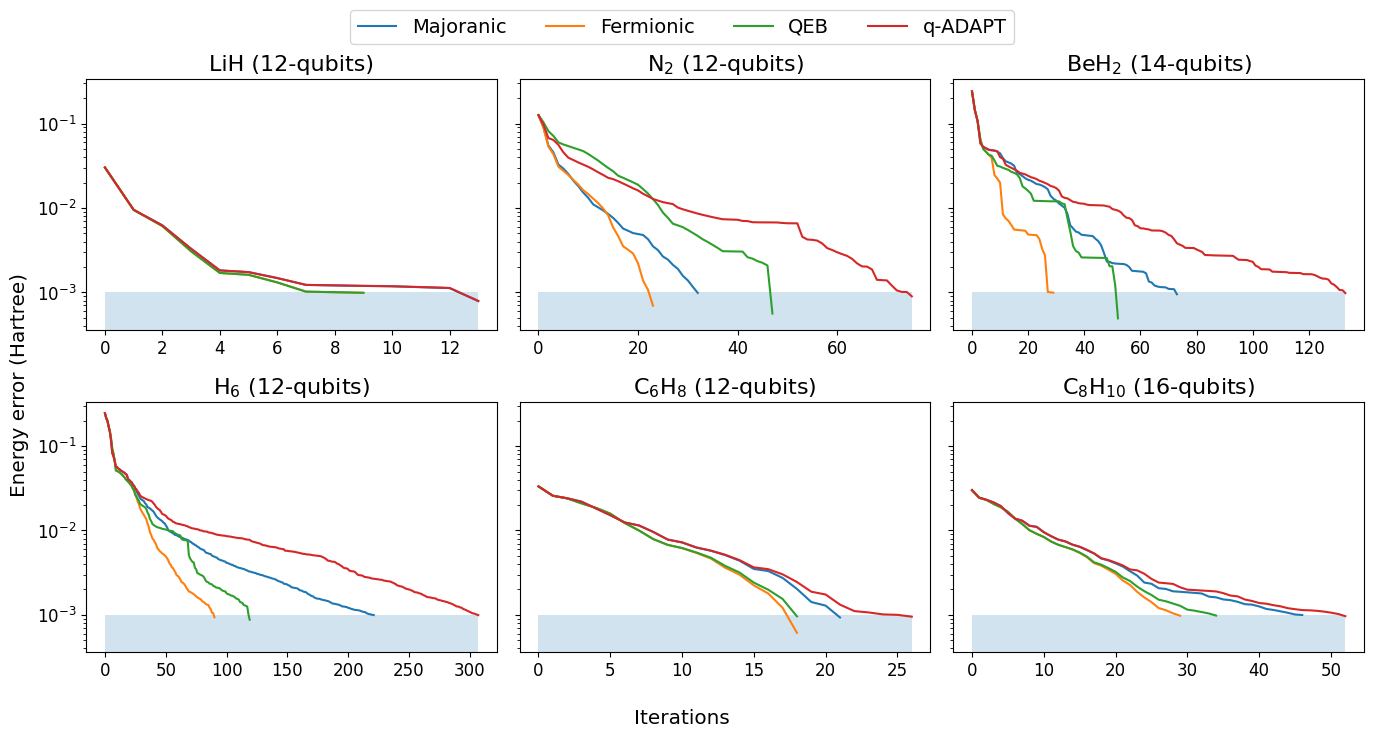}
    \caption{Convergence against ADAPT-VQE iterations to within $10^{-3}$ Hartree energy error of the exact ground-state energy. The number of iterations is equivalent to the number of variational parameters in the ansatz.}
    \label{fig:ADAPT convergence}
\end{figure*}

\begin{table}[h!]
    \begin{ruledtabular}
        \begin{tabular}{llrrrrr}
            Molecule & Pool & Initial & MS TC & MS PC & TC & PC \\
            \midrule
            LiH & Majoranic & 154 & 78 & 78 & 44 & \textbf{40} \\
            N$_2$ & Majoranic & 398 & 250 & 250 & \textbf{194} & \textbf{194} \\
            BeH$_2$ & Majoranic & 962 & 554 & 554 & 440 & \textbf{438} \\
            H$_6$ & Majoranic & 2398 & 2186 & 2186 & 1846 & \textbf{1836} \\
            C$_6$H$_8$ & Majoranic & 274 & 166 & 166 & \textbf{122} & 124 \\
            C$_8$H$_{10}$ & Majoranic & 548 & 422 & 422 & 352 & \textbf{350} \\
            LiH & Fermionic & 158 & 108 & 110 & \textbf{98} & 106 \\
            N$_2$ & Fermionic & 452 & 334 & 342 & \textbf{326} & 354 \\
            BeH$_2$ & Fermionic & 580 & \textbf{356} & \textbf{356} & 378 & 410 \\
            H$_6$ & Fermionic & 1660 & 1452 & 1468 & \textbf{1442} & 1538 \\
            C$_6$H$_8$ & Fermionic & 362 & \textbf{262} & 272 & 276 & 296 \\
            C$_8$H$_{10}$ & Fermionic & 564 & \textbf{436} & 444 & 474 & 502 \\
        \end{tabular}
    \end{ruledtabular}
    \caption{Full connectivity results for various configurations of the treespilation algorithm. 
    Mode Shuffling (MS) refers to a scenario where the mapping tree is fixed as JW, and the assignment of Fermionic modes to qubit modes is optimized. 
    This allows us to leverage the efficient circuit representation of the excitation operators presented in \cite{yordanov2021qubit}. 
    Pauli Cost (CP) and Transpiler Cost (TC) represent the possible cost functions employed. 
    The last two columns labelled PC and TC display the results when full treespilation is applied with the labelled cost function.
    The lowest CNOT counts are displayed in bold. \label{tab: FC results}}
\end{table}

\begin{table}[h!]
    \begin{ruledtabular}
        \begin{tabular}{llrrrrrrr}
        Molecule & Pool & Initial & MS TC & MS PC & CP TC & CP PC & NCP TC & NCP PC \\
        \midrule
        LiH & Majoranic & 198 & 95 & 142 & \textbf{72} & 80 & 80 & 80 \\
        N$_2$ & Majoranic & 519 & 335 & 428 & 274 & 286 & \textbf{264} & 298 \\
        BeH$_2$ & Majoranic & 1438 & 833 & 964 & \textbf{612} & 682 & 640 & 698 \\
        H$_6$ & Majoranic & 3681 & 3334 & 3216 & 2234 & 2352 & \textbf{2216} & 2374 \\
        C$_6$H$_8$ & Majoranic & 356 & 212 & 232 & \textbf{150} & \textbf{150} & 152 & 176 \\
        C$_8$H$_{10}$ & Majoranic & 750 & 587 & 690 & 430 & 458 & \textbf{420} & 466 \\
        LiH & Fermionic & 325 & 186 & 223 & \textbf{148} & \textbf{148} & 174 & 169 \\
        N$_2$ & Fermionic & 1001 & 696 & 852 & \textbf{641} & 706 & 689 & 710 \\
        BeH$_2$ & Fermionic & 1289 & 729 & 1012 & 937 & 870 & 949 & 754 \\
        H$_6$ & Fermionic & 3802 & \textbf{3359} & 3457 & 3386 & 3625 & 3490 & 3389 \\
        C$_6$H$_8$ & Fermionic & 804 & 532 & 642 & 528 & \textbf{523} & 559 & 636 \\
        C$_8$H$_{10}$ & Fermionic & 1227 & \textbf{914} & 1300 & 1201 & 1071 & 1291 & 1080 \\
        \end{tabular}
    \end{ruledtabular}
    \caption{Sycamore connectivity results for various configurations of the treespilation algorithm. 
    Connectivity and Non-connectivity Preserving (CP and NCP) settings, where the algorithm searches inside, and outside the space of connected subtrees of the device, respectively, are displayed. \label{tab: Sycamore results}}
\end{table}

\begin{table}[h!]
    \begin{ruledtabular}
        \begin{tabular}{llrrrrrrr}
        Molecule & Pool & Initial & MS TC & MS PC & CP TC & CP PC & NCP TC & NCP PC \\
        \midrule
        LiH & Majoranic & 271 & 133 & 162 & 104 & 108 & 104 & \textbf{102} \\
        N$_2$ & Majoranic & 665 & 476 & 496 & \textbf{366} & 350 & 368 & 368 \\
        BeH$_2$ & Majoranic & 1917 & 1183 & 1154 & 824 & 862 & 858 & \textbf{830} \\
        H$_6$ & Majoranic & 5261 & 4365 & 3826 & \textbf{2930} & 2956 & 2958 & 2938 \\
        C$_6$H$_8$ & Majoranic & 490 & 278 & 270 & 214 & \textbf{202} & 210 & 210 \\
        C$_8$H$_{10}$ & Majoranic & 943 & 792 & 740 & 566 & 598 & \textbf{566} & 610 \\
        LiH & Fermionic & 463 & 273 & 287 & 235 & \textbf{224} & 257 & 258 \\
        N$_2$ & Fermionic & 1585 & \textbf{1006} & 1149 & 1123 & 1168 & 1128 & 1162 \\
        BeH$_2$ & Fermionic & 2059 & \textbf{1013} & 1206 & 1297 & 1232 & 1175 & 1295 \\
        H$_6$ & Fermionic & 5783 & \textbf{4888} & 5891 & 5514 & 5773 & 5727 & 5432 \\
        C$_6$H$_8$ & Fermionic & 1276 & \textbf{770} & 854 & 827 & 849 & 909 & 924 \\
        C$_8$H$_{10}$ & Fermionic & 1898 & \textbf{1358} & 1403 & 1592 & 1707 & 1660 & 1909 \\
        \end{tabular}
    \end{ruledtabular}
    \caption{Eagle connectivity results for various configurations of the treespilation algorithm \label{tab: Eagle results}}
\end{table}

\begin{table}[h!]
    \caption{Transpilation passes used for different pools. 
    The Efficient Circuits pass (ECP) employs circuit representations of the Fermionic pool in the JW encoding and the QEB pools as introduced in \cite{Yordanov2020}.
    When possible, this pass is used as we found it more efficient than the \textit{TKET} pass that utilizes advanced compilation techniques from \cite{Cowtan2020AGC}.
    This includes the use of ECP for the ansatz produced by MS.
    The Qiskit pass uses the \textit{qiskit} transpiler with optimisation level 3 and is employed in all cases to allow CNOTs cancellation and map the circuit to the device on LC for QEB and qubit pools. 
    The Staircase pass involves the standard approach of compiling exponentiated Pauli strings into a ``staircase'' of CNOTs. 
    The Treespilation pass encompasses novel mapping optimization techniques described in this paper. 
    The Steiner pass is exclusively used to compile strings resulting from the LC Treespilation strategy. 
    Treespilation is not applicable for QEB and qubit pools due to the operators not having an exact Fermionic representation. \label{tab:transpilers}}
    \begin{ruledtabular}
        \begin{tabular}{lllll}
                      & \multicolumn{2}{c}{\textbf{Passes on Full Connectivity}}                                                                                                                        & \multicolumn{2}{c}{\textbf{Passes on Limited Connectivity}}                                                                                                                        \\
        \textit{Pool} & Initial                                                                    & \textit{Final}                                                                                     & \textit{Inital}                                                            & \textit{Final}                                                                                        \\ \hline
        Fermionic     & \begin{tabular}[c]{@{}l@{}}1. ECP  \\ 2. Qiskit \end{tabular}      & \begin{tabular}[c]{@{}l@{}} 1. Treespilation  \\ 2. ECP or TKET  \\ 3. Qiskit \end{tabular}                      & \begin{tabular}[c]{@{}l@{}}1. ECP\\ 2. Qiskit \end{tabular}            & \begin{tabular}[c]{@{}l@{}}1. Treespilation \\ 2. ECP or TKET  \\ 3. Qiskit \end{tabular} \\
        Majoranic     & \begin{tabular}[c]{@{}l@{}}1. Staircase \\ 2. Qiskit \end{tabular} & \begin{tabular}[c]{@{}l@{}}1. Treespilation \\ 2. Staircase \\ 3. Qiskit \end{tabular} & \begin{tabular}[c]{@{}l@{}}1. Staircase \\ 2. Qiskit \end{tabular} & \begin{tabular}[c]{@{}l@{}}1. Treespilation \\ 2. Steiner \\ 3. Qiskit \end{tabular}      \\
        QEB           & \begin{tabular}[c]{@{}l@{}}1. ECP\\ 2. Qiskit \end{tabular}            & \textit{Not applicable}                                                                            & \begin{tabular}[c]{@{}l@{}}1. ECP \\ 2. Qiskit \end{tabular}           & \textit{Not applicable}                                                                                        \\
        Qubit       & \begin{tabular}[c]{@{}l@{}}1. Staircase \\ 2. Qiskit \end{tabular} & \textit{Not applicable }                                                                                    & \begin{tabular}[c]{@{}l@{}}1. Staircase \\ 2. Qiskit \end{tabular} & \textit{Not applicable}                                                                                       
        \end{tabular}
    \end{ruledtabular}
\end{table}

\begin{table}[t]
    \small
    \setlength{\tabcolsep}{4pt} 
    \begin{ruledtabular}
    \begin{tabular}{ccccc}
        Molecule      & Basis  & \begin{tabular}[c]{@{}c@{}} Active Space\\ (No. Electrons, No. Orbitals) \end{tabular} & \begin{tabular}[c]{@{}c@{}} Groundstate energy\\ (Hartree) \end{tabular} & \begin{tabular}[c]{@{}c@{}} Cartesian Geometry \\ (Angstrom) \end{tabular}                                                              \\  \midrule
        LiH           & STO-3G &   (4, 6)                                                                                 & -8.654854                                                              & \begin{tabular}[c]{@{}l@{}}L 0.000  0.000  0.000\\ H 0.000  0.000  2.000 \\ \end{tabular}  \\ 
        & & & & \\                                                                          
        N$_2$         & STO-3G &   (6, 6)                                                                                & -11.402054                                                              & \begin{tabular}[c]{@{}l@{}}N 0.000  0.000  0.000\\ N 0.000  0.000  1.098 \\ \end{tabular}  \\ 
        & & & & \\                                                                          
        BeH$_2$       & STO-3G &   (6, 7)                                                                                     & -17.006486                                                         & \begin{tabular}[c]{@{}l@{}}H 0.000  0.000  0.000\\ B 0.000  0.000  2.700 \\ H 0.000  0.000  5.400\end{tabular}              \\ 
        & & & & \\                                                                          
        H$_6$         & STO-3G &   (6, 6)                                                                                    & -6.064793                                                           & \begin{tabular}[c]{@{}l@{}}H 0.000  0.000  0.000\\ H 0.000  0.000  1.500 \\ H 0.000  0.000  3.000 \\ H 0.000  0.000  4.500 \\ H 0.000  0.000  6.000 \\ H 0.000  0.000  7.500 \\ \end{tabular}              \\ 
        & & & & \\                                                                          
        C$_6$H$_8$    & cc-pVDZ &   (6, 6)                                                                                    & -5.797646                                                           & \begin{tabular}[c]{@{}l@{}}H 1.488 1.809 0.000 \\ C 0.401 1.860 0.000 \\ C 0.196 3.053 0.000 \\ H 0.375 3.974 0.000 \\ H 1.277 3.142 0.000 \\ C 0.301 0.600 0.000 \\ H 1.389 0.645 0.000 \\ C 0.301 0.600 0.000 \\ H 1.389 0.645 0.000 \\ C 0.401 1.860 0.000 \\ C 0.196 3.053 0.000 \\ H 1.488 1.809 0.000 \\ H 0.375 3.974 0.000 \\ H 1.277 3.142 0.000 \\ \end{tabular}              \\ 
        & & & & \\                                                                          
        C$_8$H$_{10}$ & cc-pVDZ &   (8, 8)                                                                                   & -8.584465                                                            & \begin{tabular}[c]{@{}l@{}}H 1.477 3.061 0.000 \\ C 0.390 3.094 0.000 \\ C 0.226 4.279 0.000 \\ H 0.331 5.208 0.000 \\ H 1.309 4.351 0.000 \\ C 0.336 0.633 0.000 \\ H 1.379 1.849 0.000 \\ C 0.291 1.825 0.000 \\ H 1.425 0.614 0.000 \\ C 0.336 0.633 0.000 \\ C 0.291 1.825 0.000 \\ H 1.425 0.614 0.000 \\ H 1.379 1.849 0.000 \\ C 0.390 3.094 0.000 \\ C 0.226 4.279 0.000 \\ H 1.477 3.061 0.000 \\ H 0.331 5.208 0.000 \\ H 1.309 4.351 0.000 \\ \end{tabular}            
    \end{tabular}
    \end{ruledtabular}
    \caption{Data for the molecules studied in this paper.\label{tab:Molecular data}}
\end{table}

\clearpage
\end{document}